\documentclass[journal,transvt]{IEEEtran}
\usepackage{amsfonts}
\usepackage{dsfont}
\usepackage{floatrow} 
\usepackage{setspace}
\usepackage{color}
\usepackage[table]{xcolor}
\usepackage{amssymb}
\usepackage{cite}
\usepackage[cmex10]{amsmath}
\usepackage{algorithm, algpseudocode}
\usepackage{algorithmicx}
\usepackage{array}
\usepackage{mathrsfs}
\usepackage{graphicx}
\usepackage{latexsym}
\usepackage{amscd}
\usepackage{amsfonts}
\usepackage{caption}
\usepackage{amsmath,amscd,amssymb,verbatim}
\usepackage{graphics}
\usepackage{amsthm}
\usepackage[T1]{fontenc}
\usepackage[utf8]{inputenc}
\usepackage{authblk}
\usepackage{optidef}
\usepackage{mathtools}
\usepackage{soul}
\usepackage{textcomp}
\usepackage{xcolor}
\usepackage{datetime}
\usepackage{microtype}

\usepackage{blindtext}
\IEEEoverridecommandlockouts

\newtheorem{theorem}{Lemma}
\newtheorem{theorem1}{Corollary}

\begin{document}
%
\renewcommand{\thefootnote}{\fnsymbol{footnote}}
\title{Trade Reliability for Security: Leakage-Failure Probability Minimization for Machine-Type Communications in URLLC  
\thanks{This work was supported in part by  the National Key R\&D Program of China under Grant 2021YFB2900301, in part by 
the German Research
Council (DFG) through the basic research project under Grant number DFG SCHM 2643/17, 
in part by 
Federal Ministry of Education and Research Germany in the program of ''Souverän. Digital. Vernetzt.'' Joint Project 6G-RIC with project identification number 16KISK028, 6G-Life with project identification number 16KISK001K, 6G-ANNA with project identification number 16KISK097 and 16KISK103.

X. Yuan, Y. Zhu and Y. Hu are with School of Electronic Information, Wuhan University, 430072 Wuhan, China and   Chair of Information Theory and Data Analytics, RWTH Aachen University, 52074 Aachen, Germany. (Email:$yuan|zhu|hu$@inda.rwth-aachen.de). 
R. F. Schaefer is with the Chair of Information Theory and Machine Learning, Technische Universität Dresden, 01062 Dresden, Germany (Email: rafael.schaefer@tu-dresden.de). 
A. Schmeink is with the Chair of Information Theory and Data Analytics, RWTH Aachen University, 52074 Aachen, Germany (Email: schmeink@inda.rwth-aachen.de). 
*Y. Hu is the corresponding  author. } }
\author{
Yao Zhu, 
Xiaopeng Yuan, 
Yulin Hu$^*$, 
Rafael F. Schaefer,
Anke~Schmeink 
}
\maketitle
    
\vspace{-5pt}
\begin{abstract}
How to provide information security while fulfilling ultra reliability and low-latency requirements is one of the major concerns for enabling the next generation of ultra-reliable and low-latency communications service (xURLLC), specially in machine-type communications. 
In this work, we investigate the reliability-security tradeoff via defining the  leakage-failure probability, which is a metric that jointly characterizes both reliability and security performances for short-packet transmissions. We discover that the system performance can be enhanced by counter-intuitively allocating fewer resources for the transmission with finite blocklength (FBL) codes. 
In order to solve the corresponding optimization problem for the joint resource allocation, we propose an optimization framework, that leverages lower-bounded approximations for the decoding error probability in the FBL regime. We characterize the convexity of the reformulated problem and establish an efficient iterative searching method, the convergence of which is guaranteed. To show the extendability of the framework, we further discuss the blocklength allocation schemes with  practical requirements of reliable-secure performance, as well as the transmissions with the statistical channel state information (CSI). 
Numerical results verify the accuracy of the proposed approach and demonstrate the reliability-security tradeoff under various setups.  
\end{abstract}
\vspace*{-0.12cm}

\section{Introduction}
Future wireless networks are expected to support ultra-reliable and low-latency communications (URLLC), which is a key enabler for many newly emerged applications, e.g., vehicle-to-vehicle/infrastructure communications, augmented/virtual reality, and factory automation~\cite{Bennis_URLLC_intro_2018,She_URLLC_intro_2021}. One of the common features of those applications is latency-sensitive and mission-critical. In view of this, {the transmissions with URLLC are carried out via finite blocklength (FBL) codes, which are also often referred to as short-packet communications}~\cite{She_URLLC_resource_2017,Durisi_shortpacket_2016}. As a key departure from the assumption of the infinite blocklength, the transmission error can no longer be negligible even if the transmission rate is lower than the channel  capacity in FBL regime. 
Therefore, the impact of FBL codes on reliability should be carefully considered for the design framework in URLLC, which requires a fundamentally different approach from that in current wireless communication systems~\cite{Durisi_shortpacket_2016}. To tackle this issue, {the authors in~\cite{Polyanskiy_2010} derived bounds on the maximal achievable transmission rate in the FBL regime, based on which abounding works have been done to provide the FBL designs,} e.g., considering cooperative relay networks~\cite{Hu_2015}, non-orthogonal multiple access (NOMA) schemes~\cite{Ding_NOMA}, wireless power transfer~\cite{Alcaraz_WP_linear_approx_app2} and green communications~\cite{Singh_2021_EE_V_1}. 

On the other hand, the broadcast nature of wireless medium  makes information security another crucial issue in future wireless communication system design~\cite{Wang_security_6G_2020}. 
This concern becomes more critical when we shift the focus from human-centric communications to machine-type communications. For example, in the  Internet-of-Thing (IoT),  instead of being external listeners, the eavesdroppers may be other internal IoT devices that wiretap confidential information~\cite{Granjal_Iot_security_2015,Hassija_Iot_security_2019}. To address this problem, significant efforts have been devoted to enhance  security performance. In particular, compared with conventional upper-layer encryption, physical layer security (PLS) has become a promising solution for machine-type communications with massive devices due to its low complexity in the implementation and flexibility in the system design~\cite{Wu_PLS_2018}. In fact, PLS for non-URLLC networks has already been extensively investigated in various scenarios~\cite{Schaefer_PLS_2017}, e.g., in UAV-assisted networks~\cite{Sun_PLS_UAV_2019}, in cognitive radio networks~\cite{Xiang_PLS_congnitive_2019}, and in visible light communications~\cite{arfaoui_PLS_visible_2020}. 

Recently, PLS in the context of URLLC for machine-type communications has also received significant attention. 
In particular, as a more general expression of widely adopted Wyner's secrecy capacity, authors in~\cite{Yang_wiretap_2019} derived a tight bound for achievable secrecy transmission rate in FBL regime, which takes the influence of both finite blocklength codes and information leakage into account. Adopting this bound, various performance metrics are proposed. For example, the authors in~\cite{Wang_PLS_convert_2022} employ the artificial noise technique to achieve both covertness and security. The secrecy transmission rate is applied directly as the metric to be maximized and an analytical approach to obtain the global optimum and an efficient approach to obtain the sub-optimum are proposed. 
In~\cite{Wang_PLS_throughput_2019}, the average achievable secrecy throughput based on the error probability with a given information leakage is proposed as a secure metric, which is minimized in both single-antenna and multiple-antenna IoT scenarios. 
This metric is adopted and extended to UAV-assisted networks in~\cite{Wang_PLS_throughput_UAV_2020}, where the UAV's trajectory and transmit power are optimized. On the other hand, authors in~\cite{Wei_PLS_throughput_FD_2021} derive the expression of secrecy throughput in full-duplex multiuser system with the consideration of self-interference and co-channel interference. 
Moreover, the authors in~\cite{Feng_outage_2022} analyze the outage probability for secure transmission in the FBL regime, based on which the optimization framework is provided, aiming at maximizing the effective throughput. 

However, the interplay between reliability and security may be overlooked in the aforementioned works, especially for machine-type communications in the next generation of URLLC (xURLLC), since both reliability and security are the key performance indicator for xURLLC.  In particular, the transmission reliability from transmitter to both legitimate receiver and eavesdropper with FBL codes depends on the allocated resources of the transmission, e.g., blocklength or transmit power. If we reduce the allocated resources on purpose, it implies the drop of the reliability for legitimate receiver, {but the probability of information leakage, i.e., eavesdroppers correctly decoding the secrecy information}, 
also drops. In other words, as long as it hurts the eavesdropper (i.e., more secure) more than it hurts the legitimate receiver (i.e., less reliable), the overall reliable-secure performance can be improved. {Therefore, we may \emph{trade reliability for security} via counter-intuitively allocating fewer resources.} This security-reliability tradeoff in FBL regime is initially discussed in~\cite{Yang_wiretap_2019} for deriving the achievable secrecy transmission rate. However, to the best of our knowledge, it has not been addressed via a probabilistic metric to represent the reliable-secure performance with an emphasis on machine-type communications. In particular, unlike human-centric communications, the transmission with a fixed amount of information only occurs sporadically in machine-type communications. Therefore, the average secrecy rate~\cite{Wang_PLS_convert_2022} is not suitable to characterize the system performance in those scenarios. On the other hand, the outage probability with fixed transmission error probability~\cite{Feng_outage_2022} and the effective throughput with fixed information leakage probability~\cite{Wang_PLS_throughput_2019} may not fully reflect the dynamic between reliability and security in FBL regime. 

Motivated by above observations, in this work, we introduce a novel metric, leakage-failure probability (LFP), which indicates the possibility that the information fails to be decoded by the legitimate receiver or is successfully decoded by the eavesdropper. {Aiming at minimizing LFP, we formulate an optimization problem to address the fundamental tradeoff between reliability and security in the classic three-node scenario, where Eve is overhearing the communication between Alice and Bob. Due to its non-convexity, we provide an optimization framework inspired by majorization-minimization~\cite{Sun_2017_MM} and successive convex approximation~\cite{Scutari_2014_SCA} algorithms to solve the problem.} In particular, our contributions can be summarized as follows:
\begin{itemize}
    \item To reveal the overall system performance with the reliability-security tradeoff, we introduce the leakage-failure probability and characterize its expression for each transmission. {The relationship between our proposed metric and other well-known metrics is also analyzed.}  
    \item The tradeoff between reliability and security is analytically presented. To address this tradeoff, we formulate a general optimization problem to minimize the leakage-failure probability by jointly allocating the blocklength and transmit power. 
    \item {An optimization framework is proposed to solve the formulated problem via decoupling the multiplication of FBL error probability and finding its lower-bounded approximations. To tackle the non-convexity issue, we provide a novel approximation for the FBL error probability and its characterization of convexity, which is a more general result compared to our previous work in~\cite[Theorem. 1]{Zhu_convex_2022}.  
    Then, We propose an iterative search algorithm with successive approximations to solve the reformulated problem. Moreover, the convergence of the algorithm and the tightness of the approximation is ensured analytically.} 
    Then, We propose an iterative search algorithm with successive approximations to solve the reformulated problem. Moreover, the convergence of the algorithm is ensured analytically. 
    \item {To show the extensibility of our optimization framework, we further discuss the low-complexity solutions for blocklength allocation scheme with requirements of high reliable-secure thresholds,  as well as under the statistical channel state information (CSI).}           
    \item Via numerical simulations, we validate our analytical findings and evaluate the performance of the proposed approach, which is able to achieve nearly optimal solutions. To provide further insight for the resource allocation, we also demonstrate the influences of channel gain, number of Eve, as well as the packet size. Especially, we observe the reliability-security tradeoff and its potential applications under various setups that may provide guidance in the system design for PLS in machine-type communications.   
\end{itemize}

The remaining part of this paper is organized as follows. In Section II, we discuss the system model and characterize the metric. Then, in Section III, we provide the optimization framework, aiming at minimizing the leakage-failure probability. Next, we investigate the extendability in Section IV. The proposed designs are evaluated in Section V, while Section VI concludes the whole work.

\section{System Model and Leakage-Failure Probability}
\subsection{System Model}
Consider the typical transmission system with three nodes, where Alice transmits confidential information to a legitimate receiver (Bob) while the eavesdropper (Eve) attempting to overhear the legitimate user messages from Alice. However, unlike the human-centric scenarios, we consider scenarios with machine-type communications, in which the transmission only occurs sporadically with limited, but important information. 
In particular, in each communication round (upon the information arrival), {the transmission of Alice is carried out with a fixed packet size of $d$ [bits] via finite blocklength codes with the blocklength of $m$ [chn.use], where the transmission rate $r=\frac{d}{m}$ is upper-bounded by $d$ [bits/chn.use].}  
Without loss of generality, we assume the channels between Alice and Bob or Eve experience quasi-static Rayleigh fading, i.e., the channels are constant within the duration of each transmission  and vary in the next. Therefore, the channel coefficients from Alice to Bob and Eve, respectively, are denoted as $h_b=\sqrt{\xi_b}\hat{h}_b$ and $h_e=\sqrt{\xi_e}\hat{h}_e$, where $\sqrt{\xi_b}$ and $\sqrt{\xi_e}$ indicate the large-scale path loss. Moreover, $\hat{h}_b\sim \mathcal{N}(0,1)$ and $\hat{h}_e\sim \mathcal{N}(0,1)$ represent the small-scale fading, which are assumed to be independent and identically distributed. 
{We assume that the devices are with a single antenna and the perfect CSI of both Bob and Eve is available at Alice}\footnote{{Our system model can also be extended to the cases that Alice equips multiple antennas, where the proposed framework is still available. Later on, we will also discuss the cases that Eve equips multiple antennas and the perfect CSI of Eve is not available.}}. 
Then, the received signals at Bob and Eve with a given transmit power of $p$ are given by:
\begin{equation}
    \label{eq:signal_bob}
    y_b=\sqrt{p}h_bs+a_b,
\end{equation}
\begin{equation}
    \label{eq:signal_eve}
    y_e=\sqrt{p}h_es+a_e,
\end{equation}
where $s$ is the transmit signal, $a_b\sim \mathcal{CN}(0,\sigma_b^2)$ and $a_e\sim \mathcal{CN}(0,\sigma_e^2)$ are the additive white Gaussian noises at Bob and Eve, respectively. Accordingly, the SNRs can be expressed as:
\begin{equation}
    \label{eq:snr_bob}
    \gamma_b=\frac{p|h_b|^2}{\sigma_b^2},
\end{equation}
 \begin{equation}
    \label{eq:snr_eve}
    \gamma_e=\frac{p|h_e|^2}{\sigma_e^2}.
\end{equation}
\subsection{Characterization of Leakage-Failure Probability (LFP)} 

Since only a limited amount of data is transmitted in each communication round, the conventional metrics, such as secrecy transmission rate or secrecy outage probability, may not fully reveal the performance of PLS in the considered scenario. The main reason is that those metrics focus on the performance characterization for the transmissions, which are carried out continuously with multiple packets. 
In view of this, we introduce a novel metric, leakage-failure probability (LFP), $\varepsilon_{LF}$, to represent the reliable-secure performance for a single transmission. It indicates the probability of the event $Y$ that the transmitted packet is either successfully decoded by Eve (i.e., leakage occurs) or incorrectly decoded by Bob (i.e., the legitimate  transmission fails). 
Let $X_b$ and $X_e$ denote the event of correct decoding by Bob and Eve, respectively. We have $\varepsilon_{LF}=\mathcal{P}(Y=1)=\mathcal{P}(X_b=0 \cup X_e=1)$. {The details for all possible combinations of the events are listed in Table.~\ref{tab:err_LF}.} 
Recall that the transmission is carried out with blocklength $m$. Due to the FBL impact, the transmission can be erroneous even if the coding rate is lower than Shannon's capacity for both Bob and Eve. In~\cite{Polyanskiy_2010}, a tight bound of the maximal achievable transmission rate is derived with a given target decoding error probability $\bar{\varepsilon}$ in AWGN channels:
\begin{equation}
\label{eq:maximal_rate}
    r^*\approx \mathcal{C(\gamma)}-\sqrt{\frac{V(\gamma)}{m}}Q^{-1}(\bar\varepsilon),\vspace{-5pt}
\end{equation}
 where  ${\mathcal{C}}(\gamma) = {\log _2}( {1 + \gamma } )$ is the Shannon's capacity and ~$V(\gamma)$ is the channel dispersion~\cite{Chem_2015}. In the complex AWGN channel, $V(\gamma)=1- {(1+\gamma)^{-2}}$.  Moreover, $Q^{-1}(x)$ is the inverse Q-function with Q-function defined as $Q(x)=\int^\infty_x \frac{1}{\sqrt{2\pi}}e^{-\frac{t^2}{2}}dt$. 
{Additionally, for any given data size $d$, according to~\eqref{eq:maximal_rate}, the (block) error probability for a single transmission is given by:
\begin{equation}
\label{eq:error_tau}
{\textstyle
\varepsilon \!=\! {\mathcal{P}}(\gamma,\frac{d}{m},m)\! \approx\! Q\Big( {\sqrt {\frac{m}{V(\gamma)}} ( {{\mathcal{C}}(\gamma ) \!-\! r)} \ln2} \Big)\mathrm{,}}\vspace{-5pt}
\end{equation} 
where $r=\frac{d}{m}$ is the transmission rate.} 
Accordingly, the error probability of decoding the packet at Bob and Eve is expressed as $\varepsilon_b=\mathcal{P}(X_b=0)=Q\Big( {\sqrt {\frac{m}{V(\gamma_b)}} ( {{\mathcal{C}}(\gamma_b) \!-\! \frac{d}{m})} \ln2} \Big)$ and
$\varepsilon_e=\mathcal{P}(X_e=0)=Q\Big( {\sqrt {\frac{m}{V(\gamma_e)}} ( {{\mathcal{C}}(\gamma_e) \!-\! \frac{d}{m})} \ln2} \Big)$, 
respectively. Therefore, the LFP is given by:
\begin{equation}
    \begin{split}
        \label{eq:err_LF}
        \varepsilon_{LF}(\varepsilon_e,\varepsilon_b)&=1-\mathcal{P}(Y=1)=\mathcal{P}(X_b=0\cup X_e=1)\\
        &=1-(1-\varepsilon_b)\varepsilon_e=\varepsilon_b\varepsilon_e+(1-\varepsilon_e).
    \end{split}
\end{equation}
The above equality chain holds since $X_b$ and $X_e$ are two independent events. 



\begin{table*}[!t]
\centering
\vspace{-15pt}
 \begin{tabular}{||c| c| c||} 
 \hline
    & $X_e=0$ & $X_e=1$ \\ 
 \hline
 $X_b=1$ &\cellcolor{green!25} secure and reliable &\cellcolor{red!25} insecure but reliable \\ 
 \hline
 $X_b=0$ &\cellcolor{red!25} secure but unreliable &\cellcolor{red!25} insecure and unreliable \\ 
 \hline
\end{tabular}
\caption{The definition for all possible outcomes of event combinations with $X_b$ and $X_e$. Red indicates  non-preferred outcomes, the possibility of which is considered as the proposed metric, leakage-failure probability $\varepsilon_{LF}$. } 
\label{tab:err_LF}
\end{table*}

\subsection{Relationship between Leakage-failure Probability and Other Metrics} 
{In fact, there are various metrics have been proposed to characterize the reliable-secure performance for PLS. In particular, some of the most common metrics are as follows: 
\begin{itemize}
    \item \textbf{Secrecy capacity}: $\mathcal{C}_s=\mathcal{C}_b-\mathcal{C}_e$~\cite{Wyner_1975_wire}, the conventional metric based on the assumption that blocklength is infinite. It is inaccurate in the FBL regime. 
    \item \textbf{Maximal achievable secrecy rate}: $r^*_s=\mathcal{C}_s-\sqrt{\frac{V(\gamma_b)}{m}}Q^{-1}(\bar{\varepsilon}_b)-\sqrt{\frac{V(\gamma_e)}{m}}Q^{-1}(\bar{\delta})$~\cite{Yang_wiretap_2019}, the upper-bound of the achievable secrecy rate with a given error probability of Bob $\bar{\varepsilon}_b$ and a given information leakage probability $\bar{\delta}=1- \bar{\varepsilon}_e$, where $\bar{\varepsilon}_e$ is the given error probability of Eve in the classic three-node scenario). However, it may not fully represent the security of information in machine-type communications, in which there is only one transmission and the packet size is fixed. 
    \item \textbf{Outage Probability}:  $\mathcal{P}_{\text{out}}=\mathcal{P}(r^*_s \leq \frac{r}{m} | \bar{\varepsilon}_b,\bar{\varepsilon}_e)$~\cite{Feng_outage_2022}, the probability that  transmissions with the rate of $r^*_s$ violates either a given reliable constraint $\varepsilon_b\geq \bar{\varepsilon}_b$ or a given secure constraint $\varepsilon_e\leq \bar{\varepsilon}_e$. It characterizes the performance with fixed constraints, and therefore can not reveal the tradeoff between reliability and security. 
    \item \textbf{error probability with a given leakage constraint}: $\varepsilon_s=Q\big(\sqrt{\frac{m}{V(\gamma_b)}}(\mathcal{C}_s-\sqrt{\frac{V(\gamma_e)}{m}}
    Q^{-1}(1-\bar{\varepsilon}_e)-\frac{d}{m}\ln 2)\big)$~\cite{Wang_PLS_throughput_2019}, the decoding error probability of Bob with a given threshold $\bar{\varepsilon_e}$ for the error probability of Eve. This metric may be inaccurate to characterize the secure-reliable performance jointly, since the threshold is given and fixed so that the leakage is not influenced by the blocklength. 
\end{itemize}
Compared with those existing metrics, our proposed LFP takes the FBL impact into account while jointly characterizes the system performance for both reliability and security. Moreover, it focuses on the single transmission. Therefore, LFP is suitable to represent the feature of machine-type communications, in which the transmission occurs sporadically.} 

Note that the transmitted packet size $d$ is fixed. The LFP can also be interpreted as the possibility that the transmission just achieves a secrecy transmission rate of $r_s^*=\frac{d}{m}$. In particular, for any combination of $\varepsilon_e$ and $\varepsilon_b$ that achieves $\varepsilon_{LF}$, it holds that
\begin{equation}
\begin{split}
\label{eq:secrecy_rate}
    &r_s^*(\varepsilon_b,\varepsilon_e|\varepsilon_{LF})\\
    =&r^*_b(\varepsilon_b)-r^*_e(\varepsilon_e)\\
    =&\mathcal{C}(\gamma_b)-\sqrt{\frac{V(\gamma_b)}{m}}Q^{-1}(\varepsilon_b)-\mathcal{C}(\gamma_e)+\sqrt{\frac{V(\gamma_e)}{m}}Q^{-1}(\varepsilon_e)\\
    =&\mathcal{C}(\gamma_b)-\mathcal{C}(\gamma_e)-\sqrt{\frac{V(\gamma_b)}{m}}Q^{-1}(\varepsilon_b)-\sqrt{\frac{V(\gamma_e)}{m}}Q^{-1}(1-\varepsilon_e)\\
    =&\mathcal{C}_s-\sqrt{\frac{V(\gamma_b)}{m}}Q^{-1}(\varepsilon_b)-\sqrt{\frac{V(\gamma_e)}{m}}Q^{-1}(\delta).
\end{split}
\end{equation}
{Interestingly, ~\eqref{eq:secrecy_rate} coincides with the metric, secrecy transmission rate, in  ~\cite[Eq. (106)]{Yang_wiretap_2019}. In other words, although the direct transformation between them seems intractable, 
LFP $\varepsilon_{LF}$ and the secrecy transmission rate $r_s^*$ is implicitly correlated.  
To more intuitively distinguish those metrics from each other, we discuss the following example:}

\emph{\textbf{Example}: {Suppose that the channel between Alice and Bob is extremely good with $\mathcal{C}(\gamma_b)\to\infty$ while the channel between Alice and Eve is limited {so that $\mathcal{C}(\gamma_e)=(\mu+1)$, where $\frac{1}{m}\leq \mu \leq 1$. In each communication round, there is only a packet with data size of $d=1$ bit needed to be transmitted and there are sufficiently long blocklength available, i.e., $m\to\infty$.}  {Therefore, we can deduce the error probability at Bob and Eve based on~\eqref{eq:error_tau}, i.e., $\varepsilon_b=\varepsilon_e=0$. Clearly, the security capacity is infinity with $\mathcal{C}_s=\mathcal{C}(\gamma_b)-\mathcal{C}(\gamma_e)=\infty$. Also, we have an infinity secrecy transmission rate, i.e., $r^*_s=\mathcal{C}(\gamma_b)-\mathcal{C}(\gamma_e)-\sqrt{\frac{V(\gamma_b)}{m}}Q^{-1}(\varepsilon_b)-\sqrt{\frac{V(\gamma_e)}{m}}Q^{-1}(\delta)=\infty$. Then, it results in an outage probability of zero, i.e., $\mathcal{P}_{\text{out}}=\mathcal{P}(r^*_s\leq \frac{d}{m}|\varepsilon_b,\varepsilon_e)=0$. According to the outage probability, the transmission is secure. 
The similar conclusion can be drawn for the error probability with a given leakage constraint $\delta>0$ , i.e.,  $\varepsilon_s=Q(\infty)=0$.
On the contrary, LFP can be written as $\varepsilon_{LF}=1-(1-0)\cdot 0=1$, which indicates an extremely insecure transmission.} In fact, in such cases, despite of Bob successfully decoding the information, it is also obtained by Eve. In other words, the transmission is not secure at all.}} 

On one hand, this example points out the key difference in physical layer security between human-centric and machine-type communications: {In human-centric communications, the performance of 
outage probability indicates the performance of security; however, in machine-type communications, the security of information can no longer be fully presented by the quantity of maximal achievable transmission rate,}  since the amount of transmitted packet size is fixed. 

On the other hand, 
we can \emph{trade reliability for security}, since  reliability and security have opposite correlations with respect to the resources, e.g., transmit power $p$ or blocklength $m$. In particular, this can be analytically observed with the following lemma.
\begin{theorem}
\label{lemma:mono}
$\varepsilon_b$ is monotonically decreasing and $\delta$ is monotonically increasing in blocklength $m$ and transmit power $p$. 
\end{theorem}
\begin{proof}
See Appendix~\ref{app:lemma_mono}.
\end{proof}
It implies that security can be improved via counter-intuitively allocating fewer resources, even if there are sufficiently more resources available. This is due to the fact that fewer resources do decrease the reliability at Bob, but also decrease the reliability at Eve, i.e., enhancing the security. As long as the decrement for Bob is less than the decrement for Eve, the reliable-secure performance is improved. Therefore, the resource allocation needs to be carefully studied to strike a balance between reliability and security that improves the overall  secure-reliable performance. 

\subsection{Problem Formulation}
In view of this, it motivates us to propose the new metric, leakage-failure probability $\varepsilon_{LF}$, and investigate the optimization framework to characterize the fundamental tradeoff between reliability and security accordingly. In particular, we aim at maximizing the secure reliability by optimally allocating the blocklength $m$ and transmit power $p$. {  With an emphasis on addressing the reliability-security tradeoff, we consider a minimal setup with unbounded $m$ and $p$, as well as unconstrained requirements for reliability and security}\footnote{{In a later section, we also investigate the optimization problem in a more practical scenario, where the resource is upper-bounded and the individual requirements for reliability and security are constrained.}}.  Therefore, the optimization problem can be written as:
\begin{mini!}
{m,p}{\varepsilon_{LF}}
{\label{prob:original}}{}
\addConstraint{p\geq 0\label{con:P}}
\addConstraint{m\in\mathbb{N}_+.\label{con:m}}
\end{mini!}
{Clearly, Problem~\eqref{prob:original} is an integer non-convex optimization problem, which can be solved with an exhaustive search. However, this may be impractical for real-world applications due to the high computational complexity. Therefore,   
in the subsequent section, we will propose an efficient optimization framework to solve the optimization problem.} 
\begin{figure}[!t]
    \centering
\includegraphics[width=0.85\textwidth,trim=0 10 0 0]{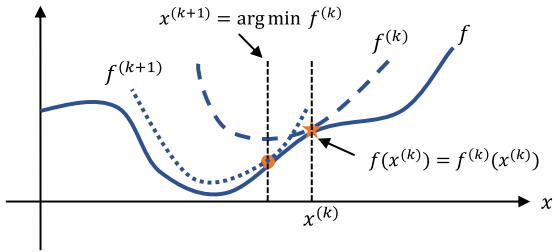}
\caption{{An example of the optimization framework for a general function $f(x)$.}
}
\label{fig:framework}
\vspace{-5pt}
\end{figure}
\section{Optimization Framework}
\label{sec:joint}
{In order to solve Problem~\eqref{prob:original}, we propose an optimization framework in this section. The general sketch of such a framework is as follows: Consider an optimization problem, in  which we aim at minimizing a general function $f(x)$ in its feasible set $x\in\mathcal{X}$ (but not necessarily convex). In the $k$-th round of iteration, we construct a lower-bounded approximation $f^{(k)}(x)\geq f(x)$, which is convex in $x$. The equality holds at the local point $x^{(k)}$, i.e., $f^{(k)}(x^{(k)})= f(x^{(k)})$. We find the minimum of $f^{(k)}$ with $x^{(k+1)}=\text{arg min} f^{(k)}(k)$, which is the local point for the next iteration. Then, we construct another lower-bounded and convex approximation $f^{(k+1)}(x)\geq f(x)$ at the $(k+1)$-th round of iterations, where it holds that $f^{(k+1)}(x^{(k+1)})= f(x^{(k+1)})$. This process repeats until the result converges. 
The convergence is ensured by the inequality chains:
\begin{equation}
\begin{split}
    f(x^{(k)})&=f^{(k)}(x^{(k)})\geq f^{(k)}(x^{(k+1)})\geq f(x^{(k+1)})\\
    &=f^{(k+1)}(x^{(k+1)})\geq f^{(k+1)}(x^{(k+2)})\geq \dots 
\end{split}
\end{equation}
An illustration of the optimization framework as shown in Fig.~\ref{fig:framework}.}  

\subsection{Bounded FBL Error Probability}
Since $\varepsilon_{LF}$ involves a multiplication of $\varepsilon_b\varepsilon_e$, which prevents further convexity analysis, we provide the following lemma to decouple it:
\begin{theorem}
\label{lemma:app}
For any functions $f_i(\boldsymbol{x})>0$, $\forall i\in\{1,\dots,N\}$, their product $\prod^N_{i=1}f_i(\boldsymbol{x})$ can be upper-bounded with respect to constants $\hat{f}_i$, $\forall i\in\{1,\dots,N\}$, i.e., {
\begin{equation}
    \label{eq:convex_app}
    \prod^N_{i=1}f_i(\boldsymbol{x})\leq \frac{1}{\prod^N_{i=1}\hat{F}_i}
        \left(
            \frac{
                \sum^N_{i=1}\hat{F}_if_i(\boldsymbol{x})
                }
                {N}
        \right)^N,
\end{equation}}
where $\hat{F}_i=\frac{\hat{f}_1}{\hat{f}_i}$, $\forall i\in\{1,\dots,N\}$.  
\end{theorem}
\begin{proof}
{First, we introduce $\hat{F}_i=\frac{\hat{f}_1}{\hat{f}_i}$ to regroup the product, i.e.,}
\begin{equation}
    \prod^N_{i=1}f_i(\boldsymbol{x})=\frac{\prod^N_{i=1}\hat{F}_i}{\prod^N_{i=1}\hat{F}_i}\prod^N_{i=1}f_i(\boldsymbol{x})=\frac{1}{\prod^N_{i=1}\hat{F}_i}{\prod^N_{i=1}\hat{F}_if_i(\boldsymbol{x})},
\end{equation}
where it holds $\hat{F}_if_i(\boldsymbol{x})\geq0$. Then, we can construct the inequality of arithmetic and geometric means for $\hat{F}_if_i(\boldsymbol{x})$:
\begin{equation}
\begin{split}
    \sqrt[N]{\prod^N_{i=1}\hat{F}_if_i(\boldsymbol{x})}&\leq\frac{\sum^N_{i=1}\hat{F}_if_i(\boldsymbol{x})}{N}\\
    \Longleftrightarrow\prod^N_{i=1}\hat{F}_if_i(\boldsymbol{x})&\leq \left( \frac{\sum^N_{i=1}\hat{F}_if_i(\boldsymbol{x})}{N} \right)^N\\
    \Longleftrightarrow\frac{1}{\prod^N_{i=1}\hat{F}_i}\prod^N_{i=1}\hat{F}_if_i(\boldsymbol{x})&\leq \frac{1}{\prod^N_{i=1}\hat{F}_i}\left( \frac{\sum^N_{i=1}\hat{F}_if_i(\boldsymbol{x})}{N} \right)^N\\
    \Longleftrightarrow\prod^N_{i=1}f_i(\boldsymbol{x})&\leq \frac{1}{\prod^N_{i=1}\hat{F}_i}\left( \frac{\sum^N_{i=1}\hat{F}_if_i(\boldsymbol{x})}{N} \right)^N,
\end{split}
\end{equation}
which completes the proof.
\end{proof}
\emph{\textbf{Remark:} Note that equality holds if $f_i(\boldsymbol{x})=\hat{f}_i$, i.e., $\hat{F}_if_i(\boldsymbol{x})=\hat{F}_jf_j(\boldsymbol{x})$, $\forall i\neq j$. This lemma can be intuitively interpreted as follows: For any product of functions with non-negative terms, we can approximate it as the summation of the same terms with a given local point, where the accuracy is guaranteed at this point.}        

According to Lemma~\ref{lemma:app}, with a given resource allocation $(\hat{m},\hat{p})$, we have{
\begin{multline}
      \label{eq:err_FL_app}
    \varepsilon_{LF}(m,p)
    \leq \frac{\varepsilon_e(\hat{m},\hat{p})}{4\varepsilon_b(\hat{m},\hat{p})}
        \left(
            \varepsilon_b(m,p)+\frac{\varepsilon_b(\hat{m},\hat{p})}{\varepsilon_e(\hat{m},\hat{p})}\varepsilon_e(m,p)
        \right)^2\\ +(1-\varepsilon_e(m,p)).  
\end{multline}
}However, such an approximation is still not sufficient to reformulate Problem~\eqref{prob:original} into a convex one, since both  $\varepsilon_e$ and $\varepsilon_b$ are  non-convex and non-concave. To tackle this, we further aim at reconstructing the Q-function $ Q( \omega)$, where $\omega(m,p)={\sqrt {\frac{m}{V(\gamma)}} ( {{\mathcal{C}}(\gamma ) \!-\! \frac{d}{m})} \ln2} $ is the auxiliary function. In particular, we define $\omega_e={\sqrt {\frac{m}{V(\gamma_e(p))}} ( {{\mathcal{C}}(\gamma_e(p) ) \!-\! \frac{d}{m})} \ln2}$ and $\omega_b={\sqrt {\frac{m}{V(\gamma_b(p))}} ( {{\mathcal{C}}(\gamma_b(p) ) \!-\! \frac{d}{m})} \ln2}$ for Bob and Eve, respectively. Then, we can establish the following lemma:
\begin{theorem}
\label{lemma:Q_app}
For a given $\hat{\omega}\in\mathcal{R}$, the Q-function $Q(\omega)=\int^\infty_\omega \frac{1}{\sqrt{2\pi}}e^{-\frac{\omega^2}{2}}d\omega$ is bounded by
\begin{equation}
    \label{eq:Q_app}
    1-b(-\hat{\omega})e^{-a(-\hat{\omega})\omega}-c(-\hat{\omega})\leq Q(\omega)\leq b(\hat{\omega})e^{-a(\hat{\omega})\omega}+c(\hat{\omega}),
\end{equation}
where
\begin{equation}
    a(\hat{\omega})=\max\{\frac{e^{-\frac{(\hat{\omega})^2}{2}}}{\sqrt{2\pi}Q(\hat{\omega})},\hat{\omega}\}>0, \label{eq:le2_a}
\end{equation}
\begin{equation}
    b(\hat{\omega})=\frac{1}{\sqrt{2\pi}\hat a}e^{\hat a\hat{\omega}-\frac{(\hat{\omega})^2}{2}}>0, \label{eq:le2_b}
\end{equation}
\begin{equation}
    c(\hat{\omega})=Q(\hat{\omega})-\hat be^{-\hat a\hat{\omega}}\label{eq:le2_c}.
\end{equation}
Specially, the equality holds for $\omega=\hat{\omega}$.
\end{theorem}
\begin{proof}
See Appendix~\ref{app:lemma_app}.
\end{proof}

Note that {the decoding error probability of any single transmission in FBL $\varepsilon(m,p)$} is a composition of functions $Q(\omega)$ and $\omega(m,p)$, i.e., $\varepsilon(m,p)=Q(\omega(m,p))$. {For example, we have $\varepsilon_b=Q(\omega_b(m,p))$ and $\varepsilon_e=Q(\omega_e(m,p))$.} 
Therefore, suppose that there is a pre-defined resource allocation with local point $(\hat{m},\hat{p})$, we can derive the following approximation for $\varepsilon_{LF}$ according to Lemma~\ref{lemma:app} and Lemma~\ref{lemma:Q_app} {along with the approximations of $\varepsilon_b$ and $\varepsilon_e$}:
\begin{equation}
\begin{split}
\label{eq:err_LF_app}
    &\varepsilon_{LF}(m,p) \\
    \lesssim&  \frac{\varepsilon_e(\hat{m},\hat{p})}{4\varepsilon_b(\hat{m},\hat{p})}
        \left(
            \hat{\varepsilon}_b(m,p)+\frac{\varepsilon_b(\hat{m},\hat{p})}{\varepsilon_e(\hat{m},\hat{p})}\hat{\varepsilon}_e(m,p)
        \right)^2+\hat{\delta}(m,p) \\
        \triangleq&\hat{\varepsilon}_{LF}(m,p|\hat{m},\hat{p}),
\end{split}
\end{equation}
where 
\begin{equation}
    \label{eq:err_b_app}
     \hat{\varepsilon}_b=
       b(\hat{\omega})e^{-a(\hat{\omega})\omega_b}+c(\hat{\omega}), 
\end{equation}
\begin{equation}
    \label{eq:err_e_app}
    \hat{\varepsilon}_e=
        b(\hat{\omega})e^{-a(\hat{\omega})\omega_e}+c(\hat{\omega}), 
\end{equation}
\begin{equation}
    \label{eq:delta_app}
    \hat{\delta}=
    b(-\hat{\omega})e^{a(-\hat{\omega})\omega_e}+c(-\hat{\omega}). 
    \end{equation}
    {According to (19), we have $\varepsilon=\hat \varepsilon(m_0,p_0|\hat{m},\hat{p})$, if $(m_0,p_0)=(\hat{m},\hat{p})$ and $\varepsilon<\hat \varepsilon(m_0,p_0|\hat{m},\hat{p})$, if $(m_0,p_0)\neq(\hat{m},\hat{p})$. In other words, this approximation is lower-bounded and it is tight at the local point $(\hat{m},\hat{p})$.} 
\subsection{Problem Reformulation and Convexity Characterization}
    According to Lemma~\ref{lemma:app} and Lemma~\ref{lemma:Q_app}, Problem~\eqref{prob:original} with given local point $(\hat{m},\hat{p})$ can be reformulated as:  
    \begin{mini!}
        {m,p}{\hat{\varepsilon}_{LF}(m,p|\hat{m},\hat{p})}
        {\label{problem:reformulated}}{}
        \addConstraint{p\geq 0}
        \addConstraint{m\geq 0\label{con:relaxed_m}}
        \addConstraint{\hat{\varepsilon}_e\leq 1,\ \hat{\varepsilon}_b\leq 1,\label{con:hat_err_e_leq_1}}
    \end{mini!}
{where additional constraints~\eqref{con:hat_err_e_leq_1} 
ensure the approximation is feasible, since the error probability can not exceed $1$. Moreover, the blocklength $m$ is relaxed from $m\in\mathbb{N}+$ into $m\geq 0$.}  
Then, we provide the following lemma to characterize the convexity:
\begin{theorem}
\label{lemma:convex}
$\omega(m,p)={\sqrt {\frac{m}{V(\gamma(p))}} ( {{\mathcal{C}}(\gamma(p) ) \!-\! \frac{d}{m})} \ln2}$ is jointly concave in blocklength $m$ and transmit power $p$ if the following condition holds:
\begin{equation}
\label{eq:convex_condition}
\begin{split}
    r\geq \frac{-\Delta_b+\sqrt{\Delta^2_b-4\Delta_a\Delta_c}}{2\Delta_a},
\end{split}    
\end{equation}
where 
\begin{equation}
\begin{split}
\Delta_a=\frac{8+9t}{4t^2},
\end{split}
\end{equation}
\begin{equation}
\begin{split}
\Delta_b=\frac{t(6t+8)-(3t+8)\mathcal{C}(\gamma)\ln2}{4t^2\ln2},
\end{split}
\end{equation}
\begin{equation}
\begin{split}
\Delta_c=\frac{t\mathcal{C}\ln2(4-3\ln2)+t^2(\mathcal{C}\ln2-1)-4\mathcal{C}^2(\ln2)^2}{4t^2(\ln2)^2},
\end{split}
\end{equation}
with $t=\gamma^2+2\gamma$. 
\end{theorem}
\begin{proof}
See Appendix~\ref{app:lemma_convex}.
\end{proof}

\emph{\textbf{Remark 2:} Although similar analytical results are already shown in~\cite{Zhu_convex_2022}, Lemma~\ref{lemma:convex} provides a more general characterization of the joint convexity, where the condition of $\gamma\geq 1$ in~\cite{Zhu_convex_2022} is no longer necessary. This significantly improves the applicability in machine-type communications with URLLC, in which the SNR of the transmission may not always be good enough due to the dynamic nature of the wireless environments. Moreover, despite the intuitive expression, the condition~\eqref{eq:convex_condition} actually presents a relationship between transmission rate and SNR for the validity of concavity. In other words, for a given SNR, a convex feasible range of transmission rate can be found. Numerically, the condition is fulfilled if $r\geq 0.023$ [bit/chn.use]. In fact, the accuracy of~\eqref{eq:error_tau} requires the transmission rate to be sufficiently high~\cite{Lancho_why_normal_approx}. In the remainder of this paper, we assume that the condition in~\eqref{eq:convex_condition} is implicitly fulfilled for most practical FBL applications in the region of interest.  
}

According to Lemma~\ref{lemma:convex}, the convexity of Problem~\eqref{problem:reformulated} can be straightforwardly characterized via the composition rule~\cite{Boyd_2004_convex} as follow:
\begin{theorem1}
Problem~\eqref{problem:reformulated} is convex.
\end{theorem1}
\begin{proof}
The objective function $\hat{\varepsilon}_{LF}$ is a sum of exponential functions with respect to $\omega_b$ and $\omega_e$ with $a\geq0$. Therefore, according to~\eqref{eq:err_b_app},~\eqref{eq:err_e_app} and~\eqref{eq:delta_app}, $\hat{\varepsilon}_{LF}(\omega_b,\omega_e)$ is convex and decreasing in $\omega_b(m,p)$ and $\omega_e(m,p)$ while both of them are jointly concave in variables $m$ and $p$, i.e., $\hat{\varepsilon}_{LF}(\omega_b(m,p),\omega_e(m,p))$ is jointly convex in $m$ and $p$. 
Obviously, all constraints are either affine or convex. As a result, Problem~\eqref{problem:reformulated} is convex.
\end{proof}
Therefore, Problem~\eqref{problem:reformulated} can be efficiently solved with standard convex optimization tools. Recall that $\hat{\varepsilon}_{LF}$ is a lower-bounded approximation of the actual objective function $\varepsilon_{LF}$, with which the accuracy is only guaranteed at $\omega=\hat{\omega}$. Solving Problem~\eqref{problem:reformulated} once is not sufficient to provide the solutions for the original Problem~\eqref{prob:original}. 
\subsection{Iterative Searching with Successive Approximations}
\label{sec:iterative}
In view of this, we propose an iterative searching method, with which the objective function is approximated successively in each iteration. In particular, {we first initialize the pair $(m^{(0)},p^{(0)})=(m_{init},p_{init})$ for the $1$-st iteration, where $m_{init}>0$ and $p_{init}>0$.} 
{ Note that those initial pair must also be feasible for Problem (23). Otherwise, the convergence rate may be impacted.}    
Then, in the $k$-th iteration, we construct Problem~\eqref{problem:reformulated} at the $k$-th local point $(\hat{m},\hat{p})=(m^{(k)},p^{(k)})$, with which the corresponding approximated LFP is defined as  $\hat{\varepsilon}^{(k)}_{LF}(m,p)=\hat{\varepsilon}_{LF}(m,p|m^{(k)},p^{(k)})$. According to Lemma~\ref{lemma:convex}, it can be solved optimally. We denote the solutions as $(m^{(k)}_{\text{opt}},p^{(k)}_{\text{opt}})$ and put it into~\eqref{eq:err_LF} to obtain the result in the $k$-th iteration, i.e., $\varepsilon^{(k)}_{LF}(m^{(k)}_{\text{opt}},p^{(k)}_{\text{opt}})$. Next, we assign $(m^{(k)}_{\text{opt}},p^{(k)}_{\text{opt}})$ as the local point of the next iteration, i.e., $(m^{(k+1)},p^{(k+1)})=(m^{(k)}_{\text{opt}},p^{(k)}_{\text{opt}})$. This process will be repeated until the gap between two iterations is smaller than a threshold $\mu_{\text{th}}$ with $|\varepsilon^{(k)}_{LF}(m^{(k)}_{\text{opt}},p^{(k)}_{\text{opt}})-\varepsilon^{(k+1)}_{LF}(m^{(k+1)}_{\text{opt}},p^{(k+1)}_{\text{opt}})|\leq \mu_{\text{th}}$. Then, we obtain the relaxed solution $(m^*_{\mathcal{R}},p^*)=(m^{(k)}_{\text{opt}},m^{(p)}_{\text{opt}})$. 
Finally, we reconstruct the solutions  by comparing the integer neighbors of the blocklength $m^*=\underset{m\in\{\lfloor m^*_{\mathcal{R}} \rfloor ,\lceil m^*_{\mathcal{R}} \rceil\}}{\text{arg\,max}} \varepsilon_{LF}(m,p^*)$, where $\lfloor\cdot\rfloor$ and $\lceil\cdot\rceil$ are the floor and ceiling functions, respectively.    
It should be emphasized that the stop criterion is based on the actual LFP instead of its approximation. Moreover, the convergence of the iteration is ensured with the following inequality chains:
\begin{equation}
    \begin{split}
    &\varepsilon_{LF}(m^{(k)},p^{(k)})=\hat{\varepsilon}^{(k)}_{LF}(m^{(k)},p^{(k)})\geq \hat{\varepsilon}^{(k)}_{LF}(m^{(k)}_{\text{opt}},p^{(k)}_{\text{opt}})\\
    \geq& \varepsilon_{LF}(m^{(k+1)},p^{(k+1)})=\hat{\varepsilon}^{(k+1)}_{LF}(m^{(k+1)},p^{(k+1)})\\
    \geq&\hat{\varepsilon}^{(k+1)}_{LF}(m^{(k+1)}_{\text{opt}},p^{(k+1)}_{\text{opt}})\geq \dots  
    \end{split}
\end{equation}
In other words, the objective can be reduced over iterations until it eventually converges to a sub-optimal solution. 
{Recall that the approximation of $\varepsilon_{LF}$ is tight at the local point. Therefore, since our proposed algorithm always converges, the tightness of the approximation is also guaranteed at the sup-optimal solution.} 
The gap between the obtained solution and the globally optimal solution is also investigated in Section~\ref{sec:simulations}. Moreover, the pseudo-code of the above algorithm is given in Algorithm~\ref{alg:iterative}. {Since Algorithm 1 solves a convex optimization problem in each iteration, the complexity can be represented as $\mathcal{O}(\theta(L)^4)$, where $\theta$ is the iteration number and $L$ is the number of variables.} 
\begin{algorithm}[!t]
\caption{Algorithm to solve Problem~\eqref{prob:original}}\label{alg:iterative}
\begin{algorithmic}[1]
\State Initialize a feasible pair $(m^{(0)},p^{(0)})$ and $k=1$.
\State Construct $\varepsilon^{(k)}$ according to~\eqref{eq:err_b_app},~\eqref{eq:err_e_app} and~\eqref{eq:delta_app}. 
\State Solve~\eqref{problem:reformulated} according to Lemma~\ref{lemma:convex} and get $(m^{(k)}_{\text{opt}},m^{(k)}_{\text{opt}})$.
\If { $|\varepsilon^{(k)}_{LF}(m^{(k)}_{\text{opt}},p^{(k)}_{\text{opt}})-\varepsilon^{(k+1)}_{LF}(m^{(k+1)}_{\text{opt}},p^{(k+1)}_{\text{opt}})|\leq \mu_{\text{th}}$ }
    \State $(m^*_{\mathcal{R}},p^*)=(m^{(k)}_{\text{opt}},m^{(k)}_{\text{opt}})$
    \Else
        \State $(m^{(k+1)},p^{(k+1)})=(m^{(k)}_{\text{opt}},p^{(k)}_{\text{opt}})$ and $r=r+1$.
        \State \textbf{Back to Step 2.}
\EndIf
\State Obtain the integer solution via $m^*=\underset{m\in\{\lfloor m^*_{\mathcal{R}} \rfloor ,\lceil m^*_{\mathcal{R}} \rceil\}}{\text{arg\,max}} \varepsilon_{LF}(m,p^*)$
\end{algorithmic}
\end{algorithm}
\subsection{Special Cases: Multiple Eves}
In machine-type communications, it is unlikely that Alice transmits information only to one destination. Instead, there may exist multiple devices, where one of the devices is the legitimate user (Bob) while the rest of them are potential eavesdroppers (Eves). In particular, consider a multi-node scenario, where ${N}$ Eves are available with the channel gain $h_{e,n}$, $\forall n\in\{1,\dots,{N}\}$. we investigate two types of Eve, i.e., passive Eves and super Eves. 

\underline{Passive Eves:} during the transmission between Alice and Bob, there are also other receivers close to Bob, which receive the broadcast signal at the same. They are not necessarily malicious, but still able to overhear the confidential message, i.e., as passive eavesdroppers~\cite{kapetanovic_passive_eve_2015}. Due to the impact of finite blocklength codes, each of the passive Eve $n$ has a possibility to incorrectly decode the packet, the event of which $X_{e,n}$ is independent to the events of others. Let $\varepsilon_{e,n}$ denote its error probability, the LFP with \underline{p}assive Eves can be rewritten as:
\begin{equation}
\begin{split}
    \label{eq:LFP_passive_eve}
    \varepsilon_{LF,P}&=1-\mathcal{P}(X_b=1\bigcup\limits_{n=1}^{N} X_{e,n}=1)=1-(1-\varepsilon_b)\prod^N_{n=1}\varepsilon_{e,n}\\
    &=\varepsilon_b\prod^N_{n=1}\varepsilon_{e,n}+(1-\prod^N_{n=1}\varepsilon_{e,n}),
\end{split}
\end{equation}
where $\varepsilon_{e,n}= Q\Big( {\sqrt {\frac{m}{V(\gamma{e,n})}} ( {{\mathcal{C}}(\gamma_{e,n} ) \!-\! \frac{d}{m})} \ln2} \Big)$ with $\gamma_{e,n}=\frac{p|h_{e,n}|^2}{\sigma^2_e}$ according to~\eqref{eq:error_tau}. 
Due to the negative multiplication of error probability of each Eve, i.e., $1-\prod^N_{n=1}\varepsilon_{e,n}$, Lemma~\ref{lemma:app} does not hold for LFP with passive Eves, which prevents the direct application of the aforementioned optimization framework. In particular, Lemma~\ref{lemma:app} requires each function $f_i(\boldsymbol{x})$ to be non-negative. To tackle this issue, we reformulate it with the following manipulations: 
\begin{equation}
\label{eq:LFP_passive_eve_reformulated}
\begin{split}
    1-\prod^N_{n=1}\varepsilon_{e,n}
    =&1+(1-\varepsilon_{e,1})\prod^N_{n=2}\varepsilon_{e,n}-\prod^N_{n=2}\varepsilon_{e,n}\\
    =&1+(1-\varepsilon_{e,1})\prod^N_{n=2}\varepsilon_{e,n}\\
    &+(1-\varepsilon_{e,2})\prod^N_{n=3}\varepsilon_{e,n}-\prod^N_{n=2}\varepsilon_{e,n}\\&...\\
    =&\sum^{{N}}_{n=1}(1-\varepsilon_{e,n})\prod^{{N}+1}_{i=n+1}\varepsilon_{e,i}.
\end{split}
\end{equation}
In particular, to ease the notation, we define $\varepsilon_{e,{N}+1}=1$. Clearly, all the terms in~\eqref{eq:LFP_passive_eve_reformulated} is non-negative, since it always holds that $0\leq \varepsilon_{e,n}\leq 1$, $\forall n \in\{1,\dots,{N}+1\}$.  Then, according to Lemma~\ref{lemma:app} and Lemma~\ref{lemma:Q_app}, we can establish the following approximation:
\begin{equation}
\label{eq:err_FL_app_eves}
    \varepsilon_{LF,P} \lesssim \hat{\varepsilon}_b\sum^{{N}}_{n=1}\hat{\varepsilon}_{e,n}+  \sum^{{N}}_{n=1}\hat{\delta}_{n}\sum^{{N}+1}_{i=n+1}\hat{\varepsilon}_{e,i},
\end{equation}
where
$
    \hat{\varepsilon}_{e,n}=
        b(\hat{\omega})e^{-a(\hat{\omega})\omega_{e,n}}+c(\hat{\omega}), 
$ 
and
$ 
    \hat{\delta}_{n}=
    b(-\hat{\omega})e^{a(-\hat{\omega})\omega_{e,n}}+c(-\hat{\omega})
$, $\forall n$. 
Then, we  follow the methodology in Section~\ref{sec:iterative} to solve the optimization problem with passive Eves by replacing~\eqref{eq:err_FL_app} with~\eqref{eq:err_FL_app_eves} in the optimization problem. To avoid repetition, we omit the details.

\underline{Super Eve:} However, if some of those receivers are malicious, e.g., as internal attackers, they may collude with each others. Consider a worst case, where each Eve is able to share their received signal perfectly so that  maximum-ratio combining (MRC) can be applied~\cite{ahn_MRC_2009,Wang_PLS_convert_2022}. Then, all ${N}$ collusive Eves with channel gain of $z_{e,n}$ can be considered as a combined super Eve with channel gain $z_{e,super}=\sum^N_{n=1}z_{e,n}$. Therefore, the optimization framework in Section~\ref{sec:joint} can be directly applied. However, it should be emphasized that, despite  the simplicity of the analytical model, super Eve is significantly difficult to be dealt with
compared with passive Eve, which will be observed and discussed in Section~\ref{sec:simulations}.

\section{Blocklength Allocation with Reliable-Secure Requirements} 
{In Section II and III, we study a minimal setup in the optimization problem to reveal the fundamental insight for the reliability-security tradeoff. For example, unconstrained reliability and security for the transmission, as well as  the assumption of the perfect CSI. Moreover, in the practical URLLC system, power control may not be an option for Alice due to the simple circuit while the blocklength allocation being in our interest. Therefore, in this section, we focus on the blocklength allocation with reliable-secure requirements, while the impact of different CSI on the system is discussed.}  
\subsection{Problem Reformulation and the Optimal Solution}
\label{sec:reliable_secure}
Consider a similar scenario as in Section~\ref{sec:joint}, where Alice transmits the message to Bob and Eve is the eavesdropper. {We still aim at minimizing the LFP $\varepsilon_{LF}$ by allocating the blocklength $m$.}  {Moreover, we also consider that the transmission has specific requirements on the reliability and security, so that $1)$ the transmission rate is larger than the Shannon's capacity of Eve so that the information leakage probability is  sufficiently low in URLLC, i.e., $\delta\leq \delta_{\max}\leq 0.5$, where $\delta_{\max}$ is the tolerance of leakage; 2) and {the transmission rate is less than the Shannon's capacity  of Bob so that the reliability is sufficiently high in URLLC}, i.e., $\varepsilon_b\leq \varepsilon_{b,\text{max}}\leq 0.5$, where $\varepsilon_{e,\text{min}}$ is a given threshold of the error probability. Then, the optimization problem~\eqref{prob:original} can be rewritten as:} 
\begin{mini!}
{m}{\varepsilon_{LF}}
{\label{prob:blocklength}}{}
\addConstraint{\delta\leq\delta_{\max}\label{con:err_e}}
\addConstraint{\varepsilon_b\leq\varepsilon_{b,\text{max}}\label{con:err_b}}
\addConstraint{m\in\mathbb{N}_+.\label{con:m_int_blocklength}}
\end{mini!}
Although the proposed optimization framework in the previous section can also be applied to solve Problem~\eqref{prob:blocklength}, in this section, we provide an approach with lower complexity for this special case. In particular, we also relax blocklength to real value, i.e., $m\in\mathbb{N}_+\to m\in\mathbb{R}_+$. Then, instead of approximating $\varepsilon_{LF}$, we directly establish the following corollary based on Lemma~\ref{lemma:convex}:
\begin{theorem}
\label{lemma:blocklength_convex}
    With the relaxation of blocklength $m\in\mathbb{R}_+$, Problem~\eqref{prob:blocklength} is convex.
\end{theorem}
\begin{proof}
First, we characterize the convexity/concavity of $\varepsilon_b$ and $\varepsilon_e$. Note that $\varepsilon_b\leq \varepsilon_{b,\text{max}}$ indicates $\mathcal{C}_b\geq d/m=r$. According to~\eqref{eq:d2w_dm2} and the composition rule of convexity, $\varepsilon_b$ is convex in $m$, which can be shown as
\begin{equation}
    \label{eq:d2err_b_dm2}
    \frac{\partial^2 \varepsilon_b}{\partial m^2}
    =\underbrace{\frac{\partial^2 \varepsilon_b}{\partial \omega^2_b}}_{\geq 0}
    \left(
        \frac{\partial \omega_b}{\partial m}
    \right)^2
    +\underbrace{\frac{\partial \varepsilon_b}{\partial \omega_b}}_{\leq 0} \underbrace{\frac{\partial^2 \omega_b}{\partial m^2}}_{\leq 0}\geq 0.
\end{equation}
It coincides with the results in~\cite{Zhu_convex_2022}. Moreover, we can deduce that constraint~\eqref{con:err_b} is also convex. On the other hand, $\varepsilon_b\leq \varepsilon_{b,\text{max}}$ implies $\mathcal{C}_b\leq d/m=r$. Therefore, the composition rule can not be applied for $\varepsilon_e$. In particular, its second derivative is given by
\begin{equation}
\begin{split}
    \frac{\partial^2 \varepsilon_e}{\partial m^2}
    =&\frac{\partial^2 \varepsilon_e}{\partial \omega^2_e}
    \left(
        \frac{\partial \omega_e}{\partial m}
    \right)^2
    +\frac{\partial \varepsilon_e}{\partial \omega_e}\frac{\partial^2 \omega_e}{\partial m^2}\\
    =&\frac{1}{4\sqrt{V_em^3}}
    \left(
        (\mathcal{C}_e+r)^2\omega_e+(\mathcal{C}_e+3r)
    \right)\\
    \overset{V_e\leq 1}{\leq}&\frac{1}{4\sqrt{V_e m^3}}
    \left(
        (\mathcal{C}_e+r)^2\sqrt{m}\underbrace{(\mathcal{C}_e-r)}_{\leq 0}+(\mathcal{C}_e+3r)
    \right)\\
    =&\frac{1}{4\sqrt{V_em^3}}
    \underbrace{-r^3-\mathcal{C}_er^2+\left(\frac{3}{m^5}+\mathcal{C}_e^2\right)+\frac{\mathcal{C}_e}{\sqrt{m}}+\mathcal{C}_e^3}_{g(r)},
\end{split}
\end{equation}
where $\mathcal{C}_e=\mathcal{C}(\gamma_e)$ and  {$V_e=V(\gamma_e)=1-(1+\gamma_e)^{-2}\leq 1$}. We show that $h(r)$ is a monotonically decreasing function with respect to $r\geq0$ via its derivative:
\begin{equation}
    \frac{\partial h(r)}{\partial r}=\left(\frac{3}{\sqrt{m^5}}+\mathcal{C}_e^2\right)-2\mathcal{C}_er-3r^2,
\end{equation}
which is a second polynomial function that opens downward and has no real root, since it holds that $\left(2\mathcal{C}_e\right)^2-4\cdot3 \left(\frac{3}{\sqrt{m^5}}+\mathcal{C}_e^2\right)\leq 0$. Recall that $0\leq \mathcal{C}_e\leq r$. In other words, if $r=0$, it must also hold that $\mathcal{C}_e=0$. Then, we have $h(r)\leq h(0)\leq 0$. It implies that 
\begin{equation}
    \frac{\partial^2 \varepsilon_e}{\partial m^2}=\frac{1}{4\sqrt{V_em^3}}h(r)\leq 0,
\end{equation}
which means that $\varepsilon_e$ is concave in $m$ and constraint~\eqref{con:err_e} is convex. 
Accordingly, we can conclude the objective function, $\varepsilon_{LF}$, is convex in $m$ by showing 
\begin{equation}
    \begin{split}
    \frac{\partial^2 \varepsilon_{LF}}{\partial m^2}
    =&\underbrace{\frac{\partial^2 \varepsilon_b}{\partial m^2}}_{\geq 0}\varepsilon_e
     +\underbrace{\frac{\partial \varepsilon_b}{\partial m}}_{\leq 0}
        \underbrace{\frac{\partial \varepsilon_e}{\partial m}}_{\leq 0}
     +\underbrace{(\varepsilon_b-1)}_{\leq 0}\underbrace{\frac{\partial^2 \varepsilon_e}{\partial m^2}}_{\leq 0}\\
    \geq&0.
     \end{split}
\end{equation}
As a result, the objective function and all constraints are convex, i.e., Problem~\eqref{prob:blocklength} is convex. 
\end{proof}
According to Lemma~\ref{lemma:blocklength_convex}, we can solve Problem~\eqref{prob:blocklength} efficiently with standard convex optimization tools. Compared with Algorithm~\ref{alg:iterative}, the convex problem only needs to be solved once with neither approximation nor iteration. Moreover, the obtained solutions are guaranteed to be the global optimum. However, we should emphasize that Lemma~\ref{lemma:blocklength_convex} is only valid with blocklength allocation under a reliable-secure requirements, which is less general than the problem considered in Section~\ref{sec:joint}.  

{Moreover, it is also worth to mention that those analytical results can be applied for other objective functions. 
For example, we can also maximize the secrecy effective throughput, which is defined as $\tau_{LF}=\frac{d}{m}(1-\varepsilon_{LF})$}\footnote{It should be emphasized that the secrecy effective throughput is different from the one in~\cite{Feng_outage_2022}, which is based on the outage probability, also different from the one in~\cite{Wang_PLS_throughput_2019}, which is based on a fixed leakage probability.}. 
{Then, we have the following modified problem: 
\begin{maxi!}
{m}{\tau_{LF}}
{\label{prob:blocklength_tau}}{}
\addConstraint{\eqref{con:err_e},\ \eqref{con:err_b}\  \text{and}\ \eqref{con:m_int_blocklength}.\nonumber}
\end{maxi!} 
We have the following Corollary to characterize it:
\begin{theorem1}
    Problem~\eqref{prob:blocklength_tau} is quasi-concave. 
\end{theorem1}
\begin{proof}
We have proven that $\varepsilon_{LF}$ is a convex function in Lemma~\ref{lemma:blocklength_convex}. Moreover, it is clear that $1/m$ is positive and convex. It means that $\tau_{LF}=d\frac{1-\varepsilon_{LF}}{m}$, which can be considered as a concave function divided by a positive and convex function, is quasi-concave~\cite{Boyd_2004_convex}.  
\end{proof}
Therefore, it exists a unique locally optimal solution in Problem~\eqref{prob:blocklength_tau}, which is the same as the globally optimal solution. It implies that Problem~\eqref{prob:blocklength_tau} can be solved efficiently, e.g., via a greedy search.} 
{The similar conclusion can be drawn for the cases with multiple transmissions. For total $T$ transmission, the effective secrecy throughput is given by $\tau=\frac{1}{T}\sum^T_{t=1}\frac{d(t)}{m}(1-\varepsilon(t))$, where $d(t)$ is the packet size at $t$-th transmission, $\forall t\in\{1,\dots,T\}$. In other words, $\tau_{LF}$ can be approximated with the sum of lower-bounded concave functions $\hat{\tau}_{LF}$, since every transmission is independent to each other.} 

\subsection{Discussions with statistical CSI of Eve}
In the previous sections, we assume that the perfect CSI of Eve is known. However, in practice, we may only know the statistical CSI instead. In particular, let $z_e=|h_e|^2$ denote the channel gain of Eve and $f_Z(z_e)$ the corresponding probability density function (PDF), the (expected) leakage-failure probability $\mathbb{E}_z[\varepsilon_{LF}]$ with the statistical CSI can be rewritten as:
\begin{equation}
    \label{eq:err_LF_int}
    \mathbb{E}_z[\varepsilon_{LF}]=\int^\infty_0\varepsilon_{LF}(m,p|z_e)f_Z(z_e)dz_e.
\end{equation}
Similarly, the error probability of Eve is given by $\mathbb{E}_z[\varepsilon_{e}]=\int^\infty_0\varepsilon_{e}(m,p|z_e)f_Z(z_e)dz_e$. 
{{2.4-3}The integral form does not influence the validity of the optimization framework of Problem~\eqref{prob:original} in Section~\ref{sec:joint}, since the convexity of the approximated LFP $\hat{\varepsilon}_{LF}$ introduced in~\eqref{eq:err_LF_app} does not depend on $z_e$. Then, $\mathbb{E}_z[\hat{\varepsilon}_{LF}]$ can be considered as a continuous sum of convex functions, and therefore we can still approximate $\mathbb{E}_z[{\varepsilon}_{LF}]$ with $\mathbb{E}_z[\hat{\varepsilon}_{LF}]$ for any given constants $(\hat{m},\hat{p})$ without influencing the convexity.} However, the validity of Lemma~\ref{lemma:blocklength_convex} has to be revisited, since the convexity/concavity of $\varepsilon_e$  only holds in a certain feasible range of $z_e$. This is due to the fact that Q-function in~\eqref{eq:error_tau} $Q(\omega)$ is a function that is first convex and then concave depending on $\omega$. To tackle this issue, we follow the methodology of our previous work~\cite{Zhu_convex_2022} to establish the convexity of $\mathbb{E}_z[\varepsilon_{LF}]$ with the statistical CSI with the following corollary: 


\begin{theorem1}
\label{corollary:fading_convex}
The convexity of $\mathbb{E}_z[\varepsilon_{LF}]$ with respect to blocklength $m$ is valid with statistical CSI of Eve.
\end{theorem1}
\begin{proof}
The second derivative of $\mathbb{E}_z[\varepsilon_{LF}]$ is given by
\begin{equation}
\begin{split}
    \frac{\partial^2 \mathbb{E}_z[{\varepsilon}_{LF}]}{\partial m^2}
    =&\underbrace{\frac{\partial^2 \varepsilon_b}{\partial m^2}}_{\geq 0}\mathbb{E}_z[\varepsilon_e]
     +\underbrace{\frac{\partial \varepsilon_b}{\partial m}}_{\leq 0}
        \underbrace{\frac{\partial \mathbb{E}_z[\varepsilon_e]}{\partial m}}_{\leq 0}
     +\underbrace{(\varepsilon_b-1)}_{\leq 0}\frac{\partial^2 \mathbb{E}_z[\varepsilon_e]}{\partial m^2}
\end{split}
\end{equation}
It is trivial to show that $\mathbb{E}_z[\varepsilon_e]$ is monotonically decreasing in $m$, i.e., $\frac{\partial \mathbb{E}_z[\varepsilon_e]}{\partial m}\leq0$, since $\varepsilon_e$ is decreasing in $m$ regardless of $z_e$. Therefore, we focus on determining whether $\mathbb{E}_z[\varepsilon_e]$ is still concave in $m$. In particular, let $\bar{z}_e$ denote the average channel gain of Eve and $\bar{z}_{\text{th}}$ is a threshold, with which  the transmission rate is just higher than the instantaneous Shannon's capacity. Recall that there is a constraint so that $\mathbb{E}_z[\varepsilon_e]\geq\varepsilon_{e,\text{min}}\geq 0.5$. It must hold that $\bar{z}_e\leq \bar{z}_{\text{th}}$. Otherwise, the constraint cannot be fulfilled. Therefore, we have
\begin{equation}
\begin{split}
\label{eq:d2e_dx2}
	\frac{\partial ^2\mathbb{E}[\varepsilon_e]}{\partial m^2} 
	=&\underbrace{{\int^{{\bar{z}_e}}_0
        \frac{\partial^2\varepsilon}{\partial x^2}f_Z(z_e)
        dz_e}}_{\leq 0}
        +\underbrace{{\int^{\bar{z}_{\text{th}}}_{{\bar{z}_e}}
            \frac{\partial^2\varepsilon}{\partial x^2}f_Z(z_e)
        dz_e}}_{\leq 0}\\
        &+\underbrace{\int^{\mathbf{\infty}}_{\bar{z}_{\text{th}}}
            \frac{\partial^2\varepsilon_e}{\partial x^2}f_Z(z_e)
        dz}_{\geq 0}.
\end{split}
\end{equation}
Moreover, since $\bar{z}_e$ is the average channel gain, the cumulative distribution function (CDF) can be broken down as follows
\begin{equation}
\begin{split}
\label{eq:ineq_1}
    \int^{\bar{z}_{\text{th}}}_{0}f_Z(z_e)dz_e&\geq \int^{\bar{z}_e}_{0}f_Z(z_e)dz_e=\int^{\infty}_{\bar{z}_e}f_Z(z_e)dz_e=\frac{1}{2}\\&\geq \int^{\infty}_{\bar{z}_{\text{th}}}f_Z(z_e)dz_e.
\end{split}
\end{equation}
Therefore, we have
\begin{equation}
    \begin{split}
        \left|
            \int^{\bar{z}_{\text{th}}}_0\frac{\partial ^2\varepsilon_e}{\partial m^2}  dz_e
        \right|
        &\geq
        \left|
            \int^{\infty}_{\bar{z}_{\text{th}}}\frac{\partial ^2\varepsilon_e}{\partial m^2}  dz_e
        \right|\\
        \Longleftrightarrow\left|
            \int^{\bar{z}_{\text{th}}}_0\frac{\partial ^2\varepsilon_e}{\partial m^2}  f_Z(z_e) dz_e
        \right|
        &\geq
        \left|
            \int^{\infty}_{\bar{z}_{\text{th}}}\frac{\partial ^2\varepsilon_e}{\partial m^2}  f_Z(z_e) dz_e
        \right|,
    \end{split}
\end{equation}
which implies that $\frac{\partial ^2\mathbb{E}[\varepsilon_e]}{\partial m^2}\leq 0$. Therefore, $\mathbb{E}[\varepsilon_e]$ is indeed concave in $m$, i.e., $\mathbb{E}[\varepsilon_{LF}]$ is convex in $m$.
\end{proof}
According to Corollary~\ref{corollary:fading_convex}, Problem~\eqref{prob:blocklength} with statistical CSI of Eve is convex. 
{Note that the validity of Corollary 3 does not depend on the distribution of $z_e$, i.e., $f_Z(z_e)$ as long as the average of $z_e$ is lower than the threshold $\bar{z}_e$. Therefore, 
Corollary 3 can also be applied for the cases where Eve equips multiple antennas. In this case, only the distribution of $z_e$ will changes and it follows gamma distribution~\cite{Feng_outage_2022}. The same conclusion can be drawn with imperfect CSI. } 
{Moreover, although we focus on the performance of single transmission in this work, Corollary 3 actually implies that we can extend the system model to multiple transmissions with fixed packet size $d$. In particular, denote $t=\{1,\dots,T\}$ is the index of total $T$ transmissions. Then, we have  $\mathbb{E}_t[\varepsilon_{LF}]=\mathbb{E}_z[\varepsilon_{LF}]$, since the time-sharing condition is satisfied~\cite{Yu_2006_dual,Zhu_2022_NOMA}. 
Therefore, we can conclude that all our previous analytical results still hold regardless of the CSI or the number of transmissions.} 
\section{Simulation Results}
\label{sec:simulations}
In this section, we validate the analytical results and investigate the system performance in the considered scenarios via numerical simulations. 
Unless specifically mentioned otherwise, we have the following setups for the simulation: We consider a normalized scenario, where we set the channel gain of Bob to $z_b=|h_b|^2=1.5$ and the channel gain of Eve to $z_e=|h_e|^2=1$. The noise power level is set as $\sigma^2_e=\sigma^2_b=0.1$. The transmitted packet size $d$ is $320$ bits. {Moreover, the available blocklength and transmit power is set to be sufficiently large with $3000$ channel uses and $10$W, respectively.} 
{The results of the integer programming for solving Problem (9) is shown as a benchmark, which provides globally optimal solutions. }

\begin{figure}[!t]
    \centering
\includegraphics[width=0.85\textwidth,trim=0 10 0 0]{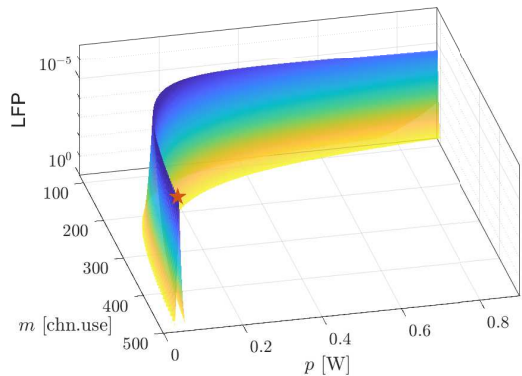}
\caption{LFP $\varepsilon_{LF}$ versus blocklength $m$ and transmit power $p$. For the sake of clarity, the points where $\varepsilon_{LF}\geq 0.5$ are omitted. The mark indicates the minimum of $\varepsilon_{LF}$. Moreover, the z-axis is reversed. }
\label{fig:LFP_vs_p_m}
\end{figure}

First, we illustrate the relationship between the LFP and its variables in Fig.~\ref{fig:LFP_vs_p_m}. In particular, we plot the values of the LFP  $\varepsilon_{LF}$ in a logarithmic scale in a reversed manner with the corresponding transmit power $p$ and blocklength $m$. For the sake of clarity, we omit the points where  $\varepsilon_{LF}\geq 0.5$. In fact, according to~\eqref{eq:err_LF}, if  $\varepsilon_{LF}\geq 0.5$, it implies that either the decoding error probability of Bob $\varepsilon_b\geq 0.5$ or the error probability of Eve $\varepsilon_e \leq 0.5$, i.e., the transmission is either unreliable or insecure. 
{As expected, $\varepsilon_{LF}$ is non-convex. It means that the standard convex method is insufficient to solve the optimization problem which involves $\varepsilon_{LF}$. This motivates us to investigate a more efficient optimization framework. 
On the other hand, with the practical leakage and error probability thresholds (corresponding to Problem~\eqref{prob:blocklength}),  $\varepsilon_{LF}$ is indeed partially convex in blocklength $m$. This phenomenon confirms our analytical findings in Lemma~\ref{lemma:blocklength_convex}.} 
However, it should be pointed out that it cannot be directly extended to the joint optimization of $m$ and $p$. This is due to the fact that the feasible set that fulfills $\varepsilon_{b}\leq \varepsilon_{b,\max}$ and $\varepsilon_{e}\leq \varepsilon_{e,\min}$, which are the reliable-secure requirements we discussed in Section~\ref{sec:reliable_secure}, is actually non-convex. 
Therefore, it is important to investigate the efficiency of our proposed optimization framework, as well as the performance gap compared with the optimal solutions, which is marked as red in Fig.~\ref{fig:LFP_k_vs_k}.

 \begin{figure}[!t]
    \centering
\includegraphics[width=0.85\textwidth,trim=0 10 0 0]{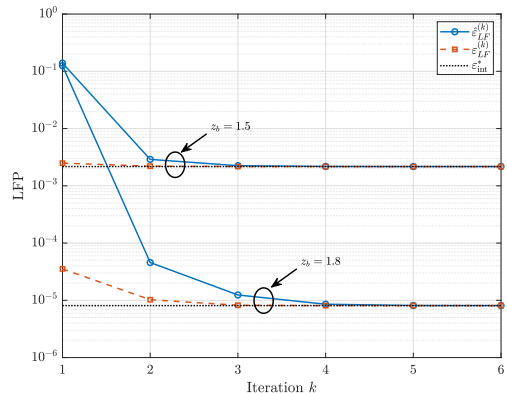}
\caption{{Obtained actual LFP $\varepsilon^{(k)}_{LF}$ and approximated $\hat{\varepsilon}^{(k)}_{LF}$ according to Algorithm 1 versus iteration $k$ under variant SNR of Bob $z_b=\{1.5, 1.8\}$. Moreover, the globally optimal results $\varepsilon^*_{int}$ obtained via integer programming are also shown as a benchmark.}}
\label{fig:LFP_k_vs_k}
\end{figure}
 
In Fig.~\ref{fig:LFP_k_vs_k}, we show the LFP obtained in each $k$-th iteration. To demonstrate the performance of the convergence for the proposed optimization framework, we plot both approximated LFP $\hat{\varepsilon}^{(k)}_{LF}$ obtained via solving the convex problem~\eqref{problem:reformulated} and actual LFP $\varepsilon^{(k)}_{LF}$ based on~\eqref{eq:err_LF}. 
{Moreover, we also show the optimal values $\varepsilon^*_{int}$, which is obtained via integer programming as the benchmark. We can observe that  $\hat{\varepsilon}^{(k)}_{LF}$ in the initial rounds are far away from the actual optimum, since it has to guarantee to be the lower-bound of the actual LFP. However, our proposed optimization framework converges quickly at a sub-linear rate. Moreover, the tightness is ensured at the converged point while  $\varepsilon^{(k)}_{LF}$ achieves the globally optimal result $\varepsilon^*_{int}$.} 
It should be mentioned that the convergence and the optimality are not influenced by the setups. For example, by varying the channel gain of Bob $z_b$ from $1.5$ to $1.8$, despite of the different values of LFP, both iterations converge at almost at the same round and achieve a (nearly) optimal solution. 
 
\begin{figure}[!t]
    \centering
\includegraphics[width=0.84\textwidth,trim=0 10 0 0]{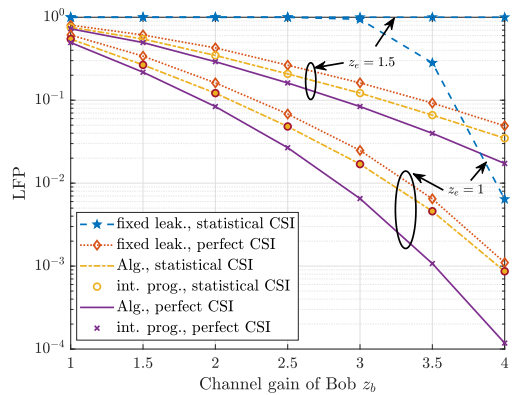}
\caption{{Minimized LFP $\varepsilon^*_{LF}$ versus channel gain of Bob $z_b$ under variant channel gains of Eve $z_e$. Both cases with perfect CSI and imperfect CSI are presented. The lines indicate solutions obtained via our proposed Algorithm~1 while the markers represent solutions obtained via integer programming. Moreover, the performance of LFP with fixed leakage threshold $\delta\leq 0.001$ is also shown as the benchmark.}}
\label{fig:LFP_vs_z_b}
\end{figure}

{
Next, we further validate the performance of our proposed approach under different setups. In Fig.~\ref{fig:LFP_vs_z_b}, we plot the minimized LFP $\varepsilon^*_{LF}$ the channel gain of Bob $z_b$ under variant channel gain of Eve, where the performance under both perfect CSI and statistical CSI is depicted, where the number of channel realizations is set as 5000. Furthermore, the minimized LFP obtained via integer programming (labeled as int. prog.), as well as via fixing the leakage threshold $\delta\leq 0.001$ (labeled as fixed leak.) is also shown as the performance benchmark. 
It should be pointed out that minimizing LFP with fixing the leakage threshold is equivalent to minimizing the error probability of Bob. Therefore, compared with Algorithm 1, the optimization problem can be more efficiently solved according to Lemma~\ref{lemma:convex}. However, its performance gap to the optimal results is also significant.  
Obviously, with the increase of $z_b$, $\varepsilon^*_{LF}$ decreases. This is not only due to the fact that larger gap between $z_b$ and $z_e$ results in higher secrecy capacity $\mathcal{C}_s=\mathcal{C}(z_b)-\mathcal{C}(z_e)$, but also thanks to the more significant tradeoff between reliability and security. In other words, via reducing each resource allocated to the transmission, we can increase $\varepsilon_e$ more, i.e., more improved security, with 
a little drop of $\varepsilon_b$, i.e., less loss of reliability. This observation implies that we should carefully allocate the available resources accordingly.}  
{More importantly, regardless of the CSI setups, similar to the results in the previous figure, the performance of our proposed algorithm is close to the performance of the integer programming. Those results also confirm the analytical findings in Section IV-B, i.e., both optimization framework and the partial convexity of LFP are valid regardless of the channel knowledge.  
Therefore, for the sake of clarity, we only show the results obtained with our algorithm in the rest of the figures.}

\begin{figure}[!t]
    \centering
\includegraphics[width=0.83\textwidth,trim=0 10 0 0]{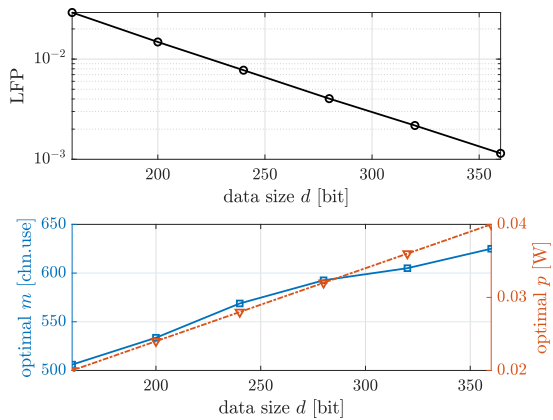}
\caption{Minimized LFP $\varepsilon^*_{LF}$ versus data size $d$ while the optimal solutions of blocklength $m^*$ and optimal transmit power $p^*$ also being depicted.}
\label{fig:LFP_vs_d_mp_vsd}
\end{figure}

According to~\eqref{eq:error_tau}, both $\varepsilon_e$ and $\varepsilon_b$ are also significantly influenced by the packet size of the transmission $d$ in a similar trend. In Fig.~\ref{fig:LFP_vs_d_mp_vsd}, we show the impact of $d$ on the system performance by plotting the minimized LFP versus data size $d$. In order to gain more insight for resource allocation schemes, we also plot the blocklength $m$ and transmit power $p$ when it achieves the minimized LFP. Interestingly, the figure shows that increasing data size $d$ actually improves the reliable-secure performance of the system. Despite that the behavior seems to be counter-intuitive, it actually fits the concept of trading reliability for security we emphasized in this work. In particular, under the condition that the channel gain of Bob is better than  channel gain of Eve, i.e., $z_b\geq z_e$, the negative influence on the decoding error probability can be compensated by allocating more radio resource, if it is available, as shown in the bottom sub-figure of Fig.~\ref{fig:LFP_vs_d_mp_vsd}. However, such  compensation is less significant for Eve, since its channel gain is worse. Therefore, a larger $d$ requires more compensation, which then causes larger performance gap between Bob and Eve in terms of decoding error probability. As a result, minimized $\varepsilon_{LF}$ decreases if we increase $d$. It should be pointed out that this phenomenon should be considered as a numerical observation instead of providing guidance for actual resource allocation schemes. {First, the packet size in machine-type communications is given, since the  amount of information is fixed. One potential solution could be combing with other technologies in the cost of the cost of implementation complexity, e.g., with advanced coding schemes or retransmission schemes to increase the redundancy.} Moreover, in practical system, the available resources are limited. To achieve the optimal solutions, a larger packet size requires more resources, which could be infeasible.

\begin{figure}[!t]
    \centering
\includegraphics[width=0.83\textwidth,trim=0 10 0 0]{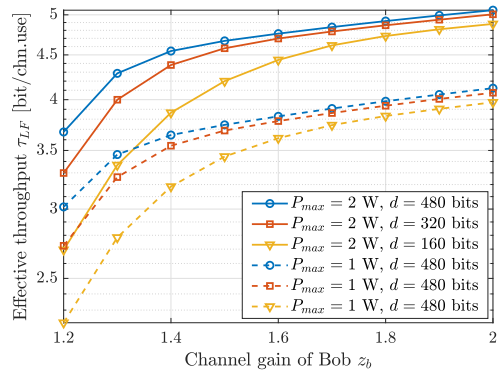}
\caption{Maximized effective throughput $\tau^*_{LF}$ versus Channel gain of Bob $z_b$ with variant maximal transmit power $P_{\max}=\{1, 2\}$ W, as well as data size $d=\{480, 320, 160\}$ bits.}
\label{fig:tau_vs_z_b_Pd}
\end{figure}

Next, we move on to investigate the system performance where we take effective throughput $\tau_{LF}$ as the metric. In  Fig.~\ref{fig:tau_vs_z_b_Pd}, we plot the maximized effective throughput $\tau^*_{LF}$ versus channel gain of Bob $z_b$. To demonstrate the different influences of available resources to $\tau^*_{LF}$ compared with that to LFP $\varepsilon^*_{LF}$, we also vary the maximal available transmit power $P_{\max}=\{1,2\}$ W and packet size $d=\{480, 320, 160\}$ bits. Similar to the LFP performance, increase $z_b$ also improve $\tau^*_{LF}$. However, the improvement becomes less significant when $z_b$ is sufficiently good. This is due to the fact that further improving $z_b$ with the high reliable-secure performance contributes little on the increment of $\tau^*_{LF}=d/m(1-\varepsilon_{LF})$, since it is dominated by $d/m$ when $(1-\varepsilon_{LF})$ approaches to 1. Therefore, 
the system prefers to choose higher transmit power $p$ with less blocklength $m$, in order to increase $\tau_{LF}$. This can be confirmed by comparing $\tau^*_{LF}$ with $P_{\max}=1$ W and $P_{\max}=2$ W. It also implies that the system is also less sensitive to the changes of packet size $d$ when $z_b$ is good enough compared with $z_e$. Under such conditions, it can always trade reliability for security, i.e., with less loss on Bob side for the high loss on Eve side. 

\begin{figure}[!t]
    \centering
\includegraphics[width=0.85\textwidth,trim=0 10 0 0]{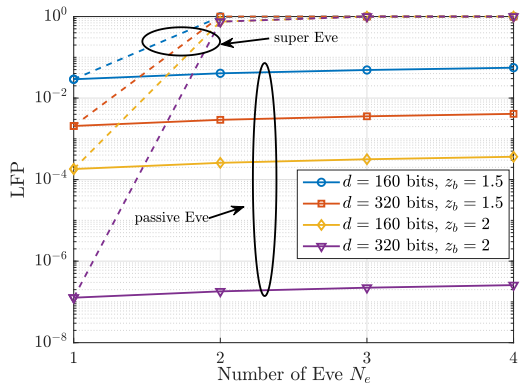}
\caption{Minimized LFP $\varepsilon^*_{LF}$ versus the number of Eve in the network ${N}$ with variant data size $d=\{160, 320\}$ [bit], as well as the channel gain of Bob $z_b=\{1.5, 2\}$. Both passive Eve and super Eve  types are considered. }
\label{fig:LFP_vs_eve}
\end{figure}

In Fig.~\ref{fig:LFP_vs_eve}, we demonstrate the influence of different types of Eve to the system by plotting the minimized LFP $\varepsilon^*_{LF}$ versus the number of Eve in the network ${N}$ in various setups of $d=\{160, 320\}$ bits and $z_b=\{1.5, 2\}$. Specially, we set $z_{e,1}=...=z_{e,{N}}=1$. Although a rise in ${N}$ causes to an increase in $\varepsilon^*_{LF}$ for both types of Eve, its impacts of them on the system performance are significantly different. On one hand, for passive Eve, we can trade the loss of Bob for the loss of each Eve, where the influences of the loss for those Eves are multiplicative with $\prod^{{N}}_{n=1}\varepsilon_{e,n}$. Therefore, ${N}$ plays a less significant role when $\varepsilon^*_{e,n}$ approaches 1. In other words, the system performance with passive Eve is dominated by those Eves with the better channel gains  while the negative effect of ${N}$ can be compromised by giving up a small amount of reliability performance. Therefore, the tolerated ${N}$ is subjected to the Eve with the best channel gain. 
On the other hand, for super Eve, each new Eve in the system introduces a higher channel gain of super Eve in an additive manner. This can not be addressed by the reliability-security tradeoff, since the loss of super Eve could be even less than the loss of Bob with higher ${N}$. In fact, the tolerated ${N}$ highly depends on the actual gap between the channel gain of Bob $z_b$ and the sum of channel gain of Eves $\sum^{{N}}_{n=1}z_{e,n}$.  

\begin{figure}[!t]
    \centering
\includegraphics[width=0.85\textwidth,trim=0 10 0 0]{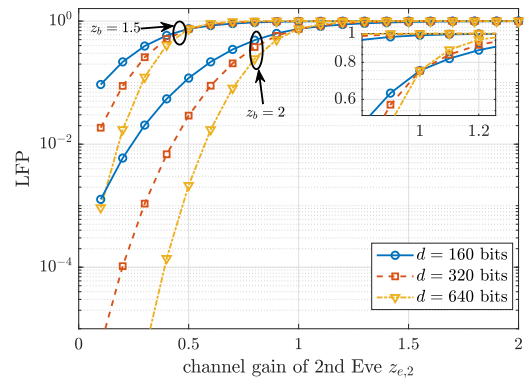}
\caption{Minimized LFP $\varepsilon_{LF}$ with two super Eve versus channel gain of second Eve $z_{e,2}$ while the channel gain of first Eve $z_{e,1}$=1 is fixed. Moreover, we vary the channel gain of Bob $z_b=\{1.5, 2\}$, as well as the packet size $d=\{160,320,640\}$ bits.} 
\label{fig:LFP_vs_sup_eve}
\end{figure}

In order to investigate the system performance for super Eve, we consider a general scenario with two Eves, where the channel gain of one Eve is fixed as $z_{e,1}=1$ while the channel gain of the second Eve $z_{e,2}$ is varied. Fig.~\ref{fig:LFP_vs_sup_eve} shows the minimized LFP $\varepsilon^*_{LF}$ under such a scenario with different setups of packet size $d=\{160,320,640\}$ bits and channel gain of Bob $z_b=\{1.5,2\}$. 
As expected, the rise of $z_{e,2}$ leads to a significant performance drop, i.e., higher $\varepsilon^*_{LF}$, regardless of setups. Similar to Fig.~\ref{fig:LFP_vs_d_mp_vsd}, we can observe that the transmission generally benefits from a larger $d$ in the cost of more allocated resources. This advantage becomes more notable when $z_b$ is elevated. 
However, this conclusion is drawn based on the condition that $z_b$ is sufficiently good so that $z_b\geq \sum^{{N}}_{n=1}z_{e,n}$. If such a condition is not fulfilled, the advantage turns into a disadvantage, with which the reliability-security tradeoff is no longer in favor of Bob (as shown in the zoom-in sub-figure of Fig.~\ref{fig:LFP_vs_sup_eve}).  In those cases, increase $d$ implies the loss of reliable-secure performance. Therefore, the robustness of the system against the super Eve is lower-bounded by $z_b$, which is generally difficult to be improved for the broadcasting channel in machine-type communications. This observation also shows the importance of other signal enhancement technologies, e.g., beam-forming technologies or non-orthogonal multiple access schemes.  

\section{Conclusion}
In this work, we investigated the approach to enhance the reliable-secure performance for machine-type communications via the resource allocation scheme in finite blocklength regime. To more accurately characterize the concept of trading reliability for security, we defined the leakage-failure probability, as the metric that jointly considering the information leakage probability and decoding error probability. We showed the relationship between such a metric and conventional secrecy transmission rate. Interestingly, we revealed that the reliable-secure performance can be enhanced by counter-intuitively allocating fewer resources for the short-packet transmission. 
In view of this, we formulated the optimization problem and provided an optimization framework accordingly. In particular, we proposed lower-bounded approximations for decoding error probability with FBL codes. Then, we established the joint convexity of those approximations in blocklength and transmit power, which can also be applied for the original error probability expression. Based on those analytical findings, we proposed an iterative search algorithm that can obtain solutions for the optimization problem, which are later shown that being able to achieve the nearly optimal solutions. With an emphasis on the extendability, we further discussed applications of our proposed optimization framework for other practical scenarios, including networks with multiple eavesdroppers and blocklength allocation with assumptions of high reliable-secure threshold. 
Via numerical simulations, we validate the performance of our proposed approaches and demonstrated the tradeoff between reliability and security while showing the influence of channel gain, data size and number of   eavesdroppers to the system. 

{ Finally, we are motivated to highlight the generality of the concept \emph{trading reliability for security}, as well as the proposed leakage-failure probability. Although we discussed this concept with the classic three-node setup, it is still available for other setups including multiple antennas, multiple access, as well as multi-functional networks. By adopting the corresponding setup combing with other practical constraints, e.g., resource budget, individual requirement of reliability or security, the leakage-failure probability can also be applied to characterize the reliable-secure performance.}



\appendices
\section{Proof of Lemma~\ref{lemma:mono}}
\label{app:lemma_mono}
\begin{proof}
Define an auxiliary function $\omega$, where $\omega(m,p)={\sqrt {\frac{m}{V(\gamma)}} ( {{\mathcal{C}}(\gamma ) \!-\! \frac{d}{m})} \ln2} $. 
It is trivial to show that $Q(\omega)=\int^\infty_\omega \frac{1}{\sqrt{2\pi}}e^{-\frac{t^2}{2}}dt$ is monotonically decreasing in $\omega$, i.e., $\frac{\partial \varepsilon}{\partial \omega}\leq 0$. Then, the monotonicity with respect to $m$ is straightforward, which can be shown via derivative test, i.e.,  
    \begin{equation}
    \frac{\partial \varepsilon}{\partial m} =\frac{\partial \varepsilon}{\partial \omega}
    \frac{\ln2}{2} m^{-\frac{1}{2}} V^{-\frac{1}{2}}(\mathcal{C}+m^{-1}d) \leq 0, 
        \end{equation}
On the other hand, the partial derivative of $\varepsilon$ with respect to $p$ is given by
\begin{equation}
\begin{split}
    \label{eq:dw_dgamma}
    \frac{\partial \varepsilon}{\partial p}=&
    \frac{\partial \varepsilon}{\partial \omega}\frac{\partial \omega}{\partial \gamma}\frac{\partial \gamma}{\partial p}\\
    =&\underbrace{|h|^2\frac{m^{\frac{1}{2}} V^{-\frac{1}{2}}}{\sigma^2\left(\gamma^{2}+2 \gamma\right)(\gamma+1)}}_{\geq 0}
    \underbrace{\left(\gamma^{2}+2 \gamma-\ln (\gamma+1)\right)}_{\Delta_1}\\
        &+\underbrace{m^{-\frac{1}{2}} V^{-\frac{3}{2}} d \ln 2 \frac{1}{(1+\gamma)^{3}}}_{\geq 0},
\end{split}
\end{equation}
where $\Delta_1(\gamma)=\gamma^{2}+2 \gamma-\ln (\gamma+1)$ is an auxiliary function with respect to $\gamma$. 
According to~\eqref{eq:dw_dgamma}, if $\Delta_1$ is non-negative, it holds that $\frac{\partial w}{\partial \gamma}\geq 0$. 
This can be proven by showing that $\Delta_1$ is monotonically increasing in $\gamma>0$ with $\frac{\partial\Delta_1}{\partial \gamma}=2\gamma+2-\frac{1}{\gamma+1}\geq 0$, i.e.,
\begin{equation}
    \Delta_1(\gamma)\geq \Delta_1(0)=0.
\end{equation}
Therefore, $\varepsilon$ is also monotonically decreasing in $p$. Immediately, we have that the monotonicity of $\delta=1-\varepsilon$ is reversed, i.e., increasing in $m$ and $p$. 
\end{proof}

\section{Proof of Lemma~\ref{lemma:Q_app}}
\label{app:lemma_app}
\begin{proof}
We first define the functions $g(\omega)$ and $h(\omega)$:
\begin{equation}
    \label{eq:lemma_g_leq_0}
    g(\omega)=Q(\omega)-b(\hat{\omega})e^{-a(\hat{\omega})\omega)}-c(\hat{\omega}),
\end{equation}
\begin{equation}
    \label{eq:lemma_h_leq_0}
    h(\omega)=Q(\omega)-1+b(-\hat{\omega})e^{-a(-\hat{\omega})\omega)}+c(-\hat{\omega}).
\end{equation}
Recall that $\hat{\omega}$ is a constant and $\omega(m,P)={\sqrt {\frac{m}{V(\gamma)}} ( {{\mathcal{C}}(\gamma ) \!-\! \frac{d}{m})} \ln2}$ is an auxiliary function with respect to blocklength $m$ and transmit power $p$. Therefore, Lemma~\ref{lemma:Q_app} can be proven by showing that $g(\omega)\leq 0$ and $h(\omega)\leq 0$ hold for arbitrary $m$ and $p$, i..e, an arbitrary $\omega$. In particular, the derivative of $g(\omega)$ with respect  to $\omega$ is given by 
\begin{equation}
\begin{split}
    \frac{\partial g}{\partial \omega}&=-\frac{1}{\sqrt{2\pi}}e^{-\frac{\omega^2}{2}}+a b e^{-a\omega}\\
    &=\frac{1}{\sqrt{2\pi}}
        \left(
            e^{a\hat{\omega}-\frac{\hat{\omega}^2}{2}-a\omega}
            -e^{-\frac{\omega^2}{2}}
        \right).
    \end{split}
\end{equation}
It is trivial to show that exponential function $e^x$ is a monotonically increasing function with respect to $x$. Therefore, since $\frac{\partial g}{\partial \omega}$ is the sum of two exponential functions, it holds that
\begin{equation}
    \text{sign}(\frac{\partial g}{\partial \omega})=\text{sign}(g_{\text{sgn}}(\omega)),
\end{equation}
where $g_{\text{sgn}}(\omega)=\frac{\omega^2}{2}-a\omega+(a\hat{\omega}-\frac{\hat{\omega}^2}{2})$ and $\text{sign}(\cdot)$ is the sign function with
\begin{equation}
    \text{sign}(x)=
    \begin{cases}
    1, & \text{if $x>0$,}\\
    0, & \text{if $x=0$,}\\
    -1, & \text{if $x<0$.}
    \end{cases}
\end{equation}
Since $g_{\text{sgn}}(\omega)$ is a quadratic polynomial function with respect to $\omega$ and $a(\omega)\geq 0$, there exists two  roots, which can be written as $\omega_1=\hat{\omega}$ and $\omega_2=2(a-\hat{\omega})$ so that $\omega_2-\omega_1=2a-2\hat \omega\geq 0$ according to~\eqref{eq:le2_a}. Therefore, the two roots are always real and 
\begin{equation}
\begin{split}
 \text{sign}(\frac{\partial g}{\partial \omega})
 =&\text{sign}(g_{\text{sgn}}(\omega))\\
&\begin{cases}
    <0,  &\text{if $\omega\in (\omega_1,\omega_2)$}, \\
    \geq0, & \text{if $\omega\in (-\infty, \omega_1]\cup[\omega_2,\infty)$}.
\end{cases}
\end{split}
\end{equation}
In other words, $g(\omega)$ is increasing in $\omega \in (-\infty,\omega_1]$, decreasing in $\omega \in (\omega_1,\omega_2)$, and then increasing again in $\omega \in [\omega_2,\infty)$, i.e., the maximum of $g(\omega)$ at $\omega=\omega_1=\hat{\omega}$ or $\omega\to\infty$. Clearly, we have $g(\hat{\omega})=0$. Moreover, according to~\eqref{eq:le2_a} and~\eqref{eq:le2_c}, it holds that
\begin{equation}
\begin{split}
    \lim_{\omega\to\infty}g(\omega)=&Q(\infty)-be^{-a\cdot \infty}-c\\
    =&1-1\frac{1}{\sqrt{2\pi}a(\hat{\omega})}-Q(\hat{\omega})< 0.
\end{split}
\end{equation}
As a result, we can deduce that
\begin{equation}
    g(\omega)\leq \max\{g(\hat{\omega}),\lim_{\omega\to\infty}g(\omega)\}\leq 0,\ \forall \omega\in\mathbb{R}. 
\end{equation}
Similarly, we can prove that $h(\omega)\leq \max \{h(-\hat{\omega}),\lim_{\omega\to-\infty}h(\omega)\}\leq 0,\ \forall \omega\in\mathbb{R}$ with $Q(x)=1-Q(-x)$ and by replacing $\hat{\omega}$ with $-\hat{\omega}$. The rest of the proof follows the above methodology. To avoid repetition, we omit the details. 
\end{proof}

\section{Proof of Lemma~\ref{lemma:convex}}
\label{app:lemma_convex}
\begin{proof}
Note that the SNR $\gamma=\frac{p|h|^2}{\sigma^2}$ is linear in the transmit power $p$. Therefore, the concavity/convexity of $\omega$ with respect to $\gamma$ is equivalent to the concavity/convexity with respect to $p$. To ease the notation, in the following, we investigate the joint concavity of $\omega$ with respect to $m$ and $\gamma$. In particular, according to Sylvester's criterion~\cite{Boyd_2004_convex}, $\omega$ can be proven as concave, if the first leading principal minor of the Hessian matrix $\mathbf{H}$, which is the second derivative $\frac{\partial^2 \omega}{\partial m^2}$, is non-negative and the determinate of Hessian matrix $\det \mathbf{H}$ is also non-negative, where 
\begin{equation} \label{eqA.2.1}
    \mathbf{H}=\left(\begin{array}{cc}
    \frac{\partial^{2} \omega}{\partial m^{2}} & \frac{\partial^{2} \omega}{\partial m \partial \gamma} \\
    \frac{\partial^{2} \omega}{\partial \gamma \partial m} & \frac{\partial^{2} w}{\partial \gamma^{2}}
    \end{array}\right)
\end{equation}
and
\begin{equation}
    \det(\mathbf{H})=\frac{\partial^2 w}{\partial m^2} \frac{\partial^2 \omega}{\partial \gamma^2}-\left(\frac{\partial^2 \omega}{\partial m \partial \gamma}\right)^2.
\end{equation}
Subsequently, we investigate each component in matrix $\mathbf{H}$. On the one hand, for the partial derivative of $w$ with respect to $m$, we have
\begin{equation}
\begin{aligned}
\label{eq:d2w_dm2}
    \frac{\partial w}{\partial m} &=\frac{\ln2}{2} m^{-\frac{1}{2}} V^{-\frac{1}{2}}(\mathcal{C}+m^{-1}d) \geq 0, \\
    \frac{\partial^2 w}{\partial m^2} &=-\frac{\ln2}{4} m^{-\frac{3}{2}} V^{-\frac{1}{2}}(\mathcal{C}+3m^{-1}d) \leq 0,
\end{aligned}  
\end{equation}
which indicates that $\omega$ is increasing and concave in $m$. 
Next, we move on to the second-order derivative, which can be written as
\begin{equation}
\begin{split}
    \label{eq:d2w_dgamma2}
    &\frac{\partial^2 w}{\partial\gamma^2}=\frac{m^{\frac{1}{2}}}{(\gamma(\gamma+2))^{\frac{5}{2}}}\\
    &\cdot\Big(
        \underbrace{-(\gamma\!+\!1)^3\!+\!\frac{1}{\gamma\!+\!1}\!+\!3(\gamma\!+\!1)\ln(\gamma\!+\!1)}_{\Delta_{2}}
        -\underbrace{3\ln 2(\gamma\!+\!1)\frac{d}{m}}_{\leq 0}
    \Big),
\end{split}
\end{equation} 
where $\Delta_2(\gamma)$ is a monotonically decreasing function, which can be shown via
\begin{equation}
\begin{aligned}
    \frac{\partial \Delta_{2}(\gamma)}{\partial \gamma} &=-3(\gamma+1)^{2}-\frac{1}{(\gamma+1)^{2}}+3+3 \ln (\gamma+1)\\ &\leq-3(\gamma+1)^{2}-\frac{1}{(\gamma+1)^{2}}+3+3 \gamma \\
    &=-3(\gamma+1) \gamma-\frac{1}{(\gamma+1)^{2}} \leq 0.
\end{aligned}
\end{equation}
Therefore, it holds that $\Delta_2(\gamma)\leq \Delta_2(0)=0$, i.e., 
\begin{equation}
    \frac{\partial^{2} \omega}{\partial \gamma^{2}}\leq\frac{m^{\frac{1}{2}}}{(\gamma(\gamma+2))^{\frac{5}{2}}}\left(0-3 \ln 2(\gamma+1) \frac{d}{m}\right) \leq 0.
\end{equation}
Hence, $\omega$ is increasing and concave in $\gamma$. Finally, we derive the expression of the last two elements in $\mathbf{H}$, which can be written as 
\begin{equation}
    \label{eq:dw_dm_dgamma}
    \frac{\partial^2 w}{\partial m \partial \gamma}=\frac{\partial^2 w}{\partial \gamma \partial m}=\frac{m^{-\frac{1}{2}} \ln 2}{2\gamma^{\frac{1}{2}}(\gamma+2)^{\frac{1}{2}}}(-\frac{\mathcal{C}(\gamma)+r}{\gamma^{2}+2 \gamma}+\frac{1}{\ln 2}).
\end{equation}
Combing~\eqref{eq:d2w_dm2},~\eqref{eq:d2w_dgamma2} and~\eqref{eq:dw_dm_dgamma}, we reformulate $\det \mathbf{H}$ via grouping in terms of $r=\frac{d}{m}$ as shown in~\eqref{eq:hessian_detail}
\begin{figure*}
\begin{equation}
\label{eq:hessian_detail}
\begin{split}
    \det \mathbf{H}
    =&\underbrace{\frac{m^{-1}(\ln2)^2}{\gamma^2+2 \gamma}}_{\geq 0}
    \left(\underbrace{
        \frac{(\gamma^2+2 \gamma)C\ln2(4-3\ln2)+(\gamma^2+2 \gamma)^2(C\ln2-1)-4C^2(\ln2)^2}{4(\gamma^2+2 \gamma)^2(\ln2)^2}
    }_{\Delta_c} \right.\\
    &\left.+\underbrace{
        \left(
            \frac{(\gamma^2+2 \gamma)(6(\gamma^2+2 \gamma)+8)-(3(\gamma^2+2 \gamma)+8)C\ln2}{4(\gamma^2+2 \gamma)^2\ln2}
        \right)
    }_{\Delta_b\geq 0}r\right.
    \left.+\underbrace{
            \left(
                \frac{8+9(\gamma^2+2 \gamma)}{4(\gamma^2+2 \gamma)^2}
            \right)
        }_{\Delta_a\geq 0}r^2
    \right)\\
    =&\frac{m^{-1}(\ln2)^2}{\gamma^2+2 \gamma}h(r),
\end{split}
\end{equation}
\hrulefill
\end{figure*}
where $h(r)=\Delta_a r^2+\Delta_b r +\Delta_c$ is a quadratic polynomial with respect to $r$. Clearly, we have $\text{sign}(\det \mathbf{H})=\text{sign}(h(r))$. Since $\gamma\geq 0$, it is straightforward to determine $\Delta_a\geq 0$ and $\Delta_b\geq 0$. 
It implies that $h(r)$ has either no real roots or one positive root at maximum while the negative root can be discarded as $r\geq 0$. Hence, $\det \mathbf{H}$ is non-negative,  if it holds  $r\geq \frac{-\Delta_b+\sqrt{\Delta^2_b-4\Delta_a\Delta_c}}{2\Delta_a}$. As a result, $\omega$ is jointly concave in $m$ and $p$. \end{proof}

\bibliographystyle{IEEEtran.bst}
\bibliography{reference_secure}
\end{document}